\documentclass[journal,onecolumn]{IEEEtran}
\usepackage{booktabs}
\usepackage{amsmath}
\usepackage{amsmath}                
\usepackage{mathtools} 
\usepackage{amsthm}                 
\usepackage{amssymb}
\usepackage{thmtools}
\usepackage{kpfonts}
\usepackage[cal=boondox]{mathalfa}

\usepackage[geometry]{ifsym}
\usepackage{dsfont}
\usepackage{multirow}
\usepackage[mathscr]{euscript}
\usepackage{graphicx}
\usepackage{color}
\usepackage[inline]{enumitem}
\usepackage{framed}
\usepackage{float}
\usepackage{longtable}
\usepackage{cite}
\usepackage{enumitem}
\usepackage{caption} 
\usepackage{subcaption} 
\usepackage{xcolor}

\usepackage[colorlinks=true, citecolor=blue, linkcolor=blue]{hyperref}       
\usepackage[font=itshape]{quoting}
\usepackage{lscape}
\usepackage{makecell}
\usepackage[section]{placeins}

\usepackage{authblk}
\usepackage{booktabs}
\usepackage{adjustbox}

\usepackage{multicol}
\usepackage{graphicx}   
\usepackage{xstring}    
\newcommand{\nz}[1]{%
  \IfStrEq{#1}{0}{}{\IfStrEq{#1}{0.0}{}{\IfStrEq{#1}{-0}{}{\IfStrEq{#1}{-0.0}{}{#1}}}}}

\usepackage{listings}
\usepackage{comment}
\usepackage{framed} 

\usepackage{algorithm}
\usepackage{algpseudocode}

\usepackage{nomencl}
\makenomenclature

\usepackage{etoolbox}
\renewcommand\nomgroup[1]{%
  \item[\bfseries
  \ifstrequal{#1}{A}{Acronyms}{%
  \ifstrequal{#1}{I}{Indices for the RCPSP}{%
  \ifstrequal{#1}{S}{Sets \& Tuples for the RCPSP}{%
  \ifstrequal{#1}{D}{Decision Variables for the RCPSP}{%
  \ifstrequal{#1}{P}{Parameters for the RCPSP}{%
  \ifstrequal{#1}{SM}{Sets \& Tuples for MBSE}{%
  \ifstrequal{#1}{PM}{Parameters for MBSE}{%
  \ifstrequal{#1}{Z}{Indices for the HFNMCF}{%
  \ifstrequal{#1}{Y}{Sets, Tuples, Functions for the HFNMCF}{%
  \ifstrequal{#1}{W}{Decision Variables for the HFNMCF}{%
  \ifstrequal{#1}{V}{Parameters for the HFNMCF}{%
  \ifstrequal{#1}{DO}{Other Decision Variables}{}}}}}}}}}}}}%
]}

\declaretheoremstyle[%
  headfont=\bfseries,%
  headpunct={:},%
  notefont=\normalfont\bfseries,%
  notebraces={--~}{},
    qed=$\blacksquare$,
]{definitionstyle}
\theoremstyle{definition}
\declaretheorem[style=definitionstyle,name=Definition]{defn}
\declaretheorem[style=definitionstyle,name=Theorem]{thm}

\theoremstyle{definition}

\theoremstyle{plain}

\newtheorem*{rem}{Remark}

\newcommand{\liinesbigfig}[4]{\begin{figure*}[ht!]\begin{center}\includegraphics[width=#4\linewidth]{#1}\vspace
{-0.1in}\caption{#2}\label{#3}\end{center}\vspace{-0.2in}\end{figure*}}

\usepackage{setspace}
\begin{document}
%
\title{Extending Resource Constrained Project Scheduling to Mega-Projects with Model-Based Systems Engineering \& Hetero-functional Graph Theory}

\author[1]{Amirreza Hosseini }
\author[1]{Amro M. Farid}
\affil[1]{Department of Systems Engineering, Charles V. Schaefer, Jr. School of Engineering and Science, Stevens Institute of Technology, Hoboken, NJ, USA}


\date{Fall 2025}
\maketitle

\begin{abstract}
Within the project management context, project scheduling serves as an indispensable component, functioning as a fundamental tool for planning, monitoring, controlling, and managing projects more broadly. Although the resource-constrained project scheduling problem (RCPSP) lies at the core of project management activities, it remains largely disconnected from the broader literature on model-based systems engineering (MBSE), thereby limiting its integration into the design and management of complex systems. The original contribution of this paper is twofold. First, the paper seeks to reconcile the RCPSP with the broader literature and vocabulary of model-based systems engineering and hetero-functional graph theory (HFGT). A concrete translation pipeline from an activity-on-node network to a SysML activity diagram, and then to an operand net is constructed. Using this representation, it specializes the hetero-functional network minimum-cost flow (HFNMCF) formulation to the RCPSP context as a systematic means of HFGT for quantitative analysis and proves that the RCPSP is recoverable as a special case of a broader model. Secondly, on an illustrative instance with renewable and non-renewable operands, the specialized HFNMCF, while producing similar schedules, yields explicit explanations of the project states that enable richer monitoring and control. Overall, the framework preserves the strengths of the classical RCPSP while accommodating real-world constraints and enterprise-level decision processes encountered in large, complex megaprojects.
\end{abstract}

%

\vspace{-0.1in}
\section{Introduction}
Project management (PM) involves the application of knowledge, skills, tools, and techniques to project activities to plan, organize, coordinate, and control them toward achieving project objectives \cite{DING:2023:00}. Within this context, project scheduling (PS) serves as an indispensable component, functioning as a fundamental tool for planning, monitoring, controlling, and, more broadly, managing projects \cite{pellerin:2019:00}. The presence of scarce resources and precedence relationships among activities makes project scheduling inherently complex \cite{hartman:2011:00}. A well-established formalization that captures these interdependencies is the resource-constrained project scheduling problem (RCPSP), which has become a standard model in the project scheduling literature \cite{hartmann:2022:00}. The RCPSP has broad practical relevance, with applications across engineering and construction, manufacturing, software development, and logistics. Moreover, it has been proven to be NP-hard, underscoring its computational complexity \cite{blazewicz:1983:00,hartmann:2022:00, DING:2023:00}. While the RCPSP already provides a robust modeling framework, numerous extensions have been proposed to accommodate the diverse conditions encountered in practice. Since the 1990s, it has also served as a foundation for the development and testing of various mathematical and heuristic algorithms aimed at addressing its computational challenges \cite{artigues:2025:00}.

\nomenclature[A]{PM}{Project Management}
\nomenclature[A]{PS}{Project Scheduling}
\nomenclature[A]{RCPSP}{Resource Constrained Project Scheduling Problem}

Although the RCPSP lies at the core of project management activities, it remains largely disconnected from the broader literature on model-based systems engineering (MBSE), thereby limiting its integration into the design and management of complex enterprises. This lack of alignment presents a significant challenge, as it prevents the seamless incorporation of optimal project scheduling into mega-project management where enterprise architecture is integral to project delivery.  Without a cohesive framework that integrates project scheduling with MBSE and enterprise architecture, project scheduling remains an isolated consideration rather than an integral component of the engineering management of meega-projects.  This disconnect leads to missed opportunities for identifying and addressing scheduling trade-offs during the critical early phases of system planning-phases in which decisions have the greatest impact on project outcomes. As highlighted by the ``Iron Law of Project Management," more than 99.5\% of projects fail to meet their original cost, schedule, and benefit targets \cite{flyvbjerg:2014:00, Hosseini:2025:ISC-J55}. Furthermore, the inability to jointly plan and schedule projects alongside technical, operational, and economic objectives often results in fragmented decision-making, wherein RCPSP-derived schedules may conflict with broader system-level goals. Ultimately, this fragmentation hinders the development of holistic solutions that effectively balance project scheduling with the complex and multifaceted demands of modern systems, limiting progress toward achieving on-time, on-budget, and value-driven project delivery.

\nomenclature[A]{MBSE}{Model-based Systems Engineering}

\vspace{-0.1in}
\subsection{Original Contribution}
The original contribution of this paper is twofold.  First, the paper seeks to reconcile the resource-constrained project scheduling problem with the broader literature and vocabulary of model-based systems engineering\cite{Delligatti2013SysMLDA} and hetero-functional graph theory\cite{Schoonenberg:2019:ISC-BK04, Farid:2022:ISC-J49, Farid:2025:ISC-JR06}. It does so by providing a constructive translation from an Activity-on-Node network to a SysML activity diagram, and then to an operand net, and by specializing the hetero-functional network minimum-cost flow problem to the RCPSP context. Doing so reveals the specific limiting conditions upon which model-based systems engineering (MBSE) and hetero-functional graph theory (HFGT) collapse to project scheduling.  Therefore, this paper proves that model-based systems engineering and hetero-functional graph theory are a formal generalization of the resource-constrained project scheduling problem.  Consequently, this paper discusses how the resource-constrained project scheduling problem can be extended when these specific limiting conditions are relaxed.  In particular, it shows how MBSE and HFGT address projects and megaprojects that:
\begin{itemize}
    \item involve inherent complexity, including intricate precedence relationships and dependencies,
    \item require the management of multiple resource types, each subject to distinct and potentially restrictive constraints,
    \item and are amenable to systematic extension through comprehensive taxonomies, enabling the incorporation of a wide range of real-world conditions and project scenarios.
\end{itemize}
Secondly, the paper demonstrates how model-based systems engineering and hetero-functional graph theory may be used to enhance project scheduling for the complex real-life constraints found in mega-projects.  

\nomenclature[A]{HFGT}{Hetero-functional Graph Theory}

\vspace{-0.1in}
\subsection{Paper Outline}
The remainder of the paper proceeds as follows.  Section \ref{Sec:Background} provides preliminary background on the resource-constrained project scheduling problem, model-based systems engineering, and hetero-functional graph theory.  Section \ref{Sec:Reconciliation} then reconciles these three modeling and analysis techniques.  Ultimately, it shows that MBSE and HFGT are a formal generalization of the resource-constrained project scheduling problem.  Sec. \ref{Sec:Discussion} then discusses how this MBSE and HFGT can be used to account for the complex real-life constraints found in mega-projects.  Finally, Sec. \ref{Sec:Conclusion} brings the work to a conclusion.  

\vspace{-0.1in}
\section{Background}\label{Sec:Background}

In order to support the analytical discussion in the following sections, this section provides preliminary background on the resource-constrained project scheduling problem in Sec. \ref{Sec:RCPSP-Background}, on model-based systems engineering in Sec. \ref{Sec:MBSEIntro}, and on hetero-functional graph theory in Sec. \ref{Sec:HFGTIntro}.  The reader is referred to the Nomenclature section in the Appendix for a transparent listing of all mathematical symbols and their meanings as they pertain to each of these subsections.

\vspace{-0.1in}
\subsection{Resource-constrained Project Scheduling Problem}\label{Sec:RCPSP-Background}
In general, the resource-constrained project scheduling problem assumes an Activity-on-Nodes project network $G=\{{\cal V},{\cal A}\}$ where the vertices ${\cal V}$ represent activities and the directed arcs ${\cal A}$ represent precedence between them.  It includes three fundamental constraints among project activities: precedence constraints, task durations, and competition for limited resources \cite{hartmann:2022:00}. Formally, the problem involves determining a feasible assignment of non-preemptive start times to a set of activities ${\cal V}=\{{\cal v}_1, \ldots, {\cal v}_{n}\}$.  Additionally, a dummy start activity ${\cal v}_o$ and dummy end activity ${\cal v}_f$ are added to create an augmented activity set ${\cal V}_A=\{{\cal v}_o, {\cal V}, {\cal v}_f\}$.  Each activity requires one or more resources $s_{\cal l} \in S_{\cal R}$, each with its availability $C_{\cal lR}$.  The objective is to minimize the project makespan; the earliest possible completion of the terminal activity $v_f$. Due to its combinatorial complexity, the RCPSP is strongly NP-hard \cite{blazewicz:1983:00}. Over time, numerous extensions have been proposed to address practical real-world conditions, including preemptive scheduling \cite{moukrim:2015:00,shahabi:2024:00}, time-varying resource demands \cite{hartmann:2014:00, pouramin:2024:00}, setup and transfer times \cite{ma:2022:00, dashti:2025:00}, multi-modal execution with renewable and nonrenewable resources \cite{yang:2024:00, jarboui:2008:00, peng:2025:00}, and trade-offs in time, cost, and quality \cite{yazdani:2024:00, maqsoodi:2023:00, polancos:2024:00}. Temporal and logical extensions encompass generalized precedence relations \cite{goudarzi:2024:00, karnebogen:2024:00}, time windows, overlapping activity logic, and conditional or optional activity networks \cite{yu:2025:00}. Additionally, models have expanded to incorporate diverse resource structures such as cumulative, multi-skilled, continuous, or partially renewable resources \cite{goudarzi:2024:00, polancos:2024:00, maghsoudlou:2017:00, karnebogen:2024:00}.  Objectives in some works are robustness, resource investment, and financial performance measures like net present value, along with algorithm hybridization and development \cite{hussain:2024:00,qiu:2025:00, Liu:2025:00, Balouka:2016:00, mohammadagha:2025:00, Peng:2023:00}. This rich body of extensions has transformed the RCPSP from a stylized academic construct into a flexible framework capable of capturing the complexities of real-world mega-projects.

\nomenclature[A]{AoN}{Activity-on-Node}
\nomenclature[S]{${\cal V}$}{Set of AoN network vertices}
\nomenclature[S]{${\cal a_{ij} \in A}$}{Set of AoN network edges(arcs)}
\nomenclature[S]{${\cal V}_A$}{Set of RCPSP activities}
\nomenclature[S]{${\cal v_o}$}{Dummy start activity}
\nomenclature[S]{${\cal v_f}$}{Dummy terminal activity}
\nomenclature[S]{$s_{\cal l} \in S_{\cal R}$}{RCPSP Resources (Required Operands)}
\nomenclature[P]{$C_{\cal {lRr}}$}{Capacity of renewable resource ${\cal l}$}
\nomenclature[P]{$C_{\cal {lRn}}$}{Capacity of non-renewable resource ${\cal l}$}

Given this vast diversity of variations, the RCPSP is best viewed as a class of problems rather than a single optimization problem. To that end, a taxonomy of RCPSP problems has been developed as part of a larger taxonomy of machine scheduling problems, which consists of three fields: $\alpha|\beta|\gamma$~\cite{demeulemeester:2002:00, Brucker:2007:00,kolisch:2013:00}. Similar to machine scheduling, $\alpha$ describes the machine/resource environment, $\beta$ characterizes the activities, and $\gamma$ specifies the performance measures.
\nomenclature[P]{$\alpha$}{Machine/resource environment of scheduling problems}
\nomenclature[P]{$\beta$}{Activity/job characteristic of scheduling problems}
\nomenclature[P]{$\gamma$}{Optimality criteria and performance measure of scheduling problems}

In the alpha branch of the machine scheduling taxonomy, the $\alpha$-field can take at most four values. First, $\alpha_1$ describes the structural resources of the production process. As noted in \cite{demeulemeester:2002:00, Brucker:2007:00}, the RCPSP requires $\alpha_1=\emptyset$ and the remaining characteristics of the machine/resource environment are specified in the remaining $\alpha$ fields. Second, $\alpha_2$ describes the number of resources ($\alpha_2 \in \{o,1,m\}$). Third, $\alpha_3$ denotes the specific resource types ($\alpha_3 \in \{o, 1, T, 1T, v, \chi, \sigma\}$). This field provides information on whether resources are renewable, cumulative, spatial, or other types as elaborated in \cite{demeulemeester:2002:00, Brucker:2007:00}. Finally, $\alpha_4$ describes the resource availability ($\alpha_4 \in \{o,k,va,\textbf{a}, \tilde{a}, \textbf{va}, v\tilde{a}\}$). This field indicates whether the resource availability is arbitrary, constant, fuzzy, variable, stochastic, etc. The interested reader is referred to  \cite{demeulemeester:2002:00, Brucker:2007:00} for further information.  

In the beta branch of the taxonomy, the second field, $\beta$, specifies the activity characteristics of a project scheduling problem. It can incorporate up to nine elements:
\begin{enumerate}
    \item $\beta_1$ indicates the possibility of activity preemption ($\beta_1 \in \{o, \text{pmtn}, \text{pmtn-rep}\}$). Preemption means that an activity can be interrupted and resumed later \cite{demeulemeester:2002:00, Brucker:2007:00}.
    \item $\beta_2$ indicates the precedence relations. This field may include strict finish-start (FS) constraints, as found in CPM and PERT networks, and is denoted by $\beta_2 = cpm$.  It also includes different types of generalized rules and probabilistic networks \cite{demeulemeester:2002:00,Brucker:2007:00}.
    \item $\beta_3$ describes the ready time before each activity can start ($\beta_3 \in \{o, \rho_j, \boldsymbol{\rho_j},\tilde{\rho_j}\}$)\cite{demeulemeester:2002:00,Brucker:2007:00}.
    \item $\beta_4$ specifies the duration of project activities ($\beta_4 \in \{o, \text{cont}, (d_j = d), \boldsymbol{d_j}, \tilde{d_j}\}$). Activity durations may be integer, continuous, or stochastic\cite{demeulemeester:2002:00,Brucker:2007:00}.
    \item $\beta_5$ indicates whether there are deadlines on activity $j$ or the overall project ($\beta_5 \in \{o, \delta_j,\delta_n\}$)\cite{demeulemeester:2002:00,Brucker:2007:00}.
    \item $\beta_6$ describes the resource requirements of activities ($\beta_6 \in \{o, k, \text{vr}, \text{disc}, \text{cont}, \text{int}, \text{conc}, \text{conv}, \text{lin}\}$)\cite{demeulemeester:2002:00,Brucker:2007:00}.
    \item $\beta_7$ indicates whether activities can be executed in multiple modes ($\beta_7 \in \{o, \text{mu}, \text{id}\}$)\cite{demeulemeester:2002:00,Brucker:2007:00}.
    \item $\beta_8$ describes the financial aspects of activities. This includes whether cash flows exist in the scheduling problem, and whether they are activity-based, time-based, and fuzzy, variable, or deterministic \cite{demeulemeester:2002:00,Brucker:2007:00}.
    \item $\beta_9$ denotes whether there are changeover times and whether they are stochastic, fuzzy, or deterministic ($\beta_9 \in \{ o, s_{jk}, \boldsymbol{s_{jk}}, \tilde{s}_{jk}\}$)\cite{demeulemeester:2002:00,Brucker:2007:00}.
\end{enumerate}

In the gamma branch of the taxonomy, the third field, $\gamma$, describes the optimality criteria and elaborates on the objective function of the problem. In general, objectives can be categorized into two main types: \begin{enumerate*}
    \item early completion measures, which aim to minimize the project makespan in various forms, and
    \item free completion measures, which do not seek to minimize the project duration but instead focus on other goals such as net present value (NPV).
\end{enumerate*} For example, $\gamma = C_{max}$ denotes a makespan minimization objective; $\gamma = T_{max}, L_{max}$ represent the minimization of tardiness and lateness, respectively; and $\gamma = \sum \text{sq.dev}$ indicates the objective of minimizing the sum of squared deviations of resource usage from the average.

\liinesbigfig{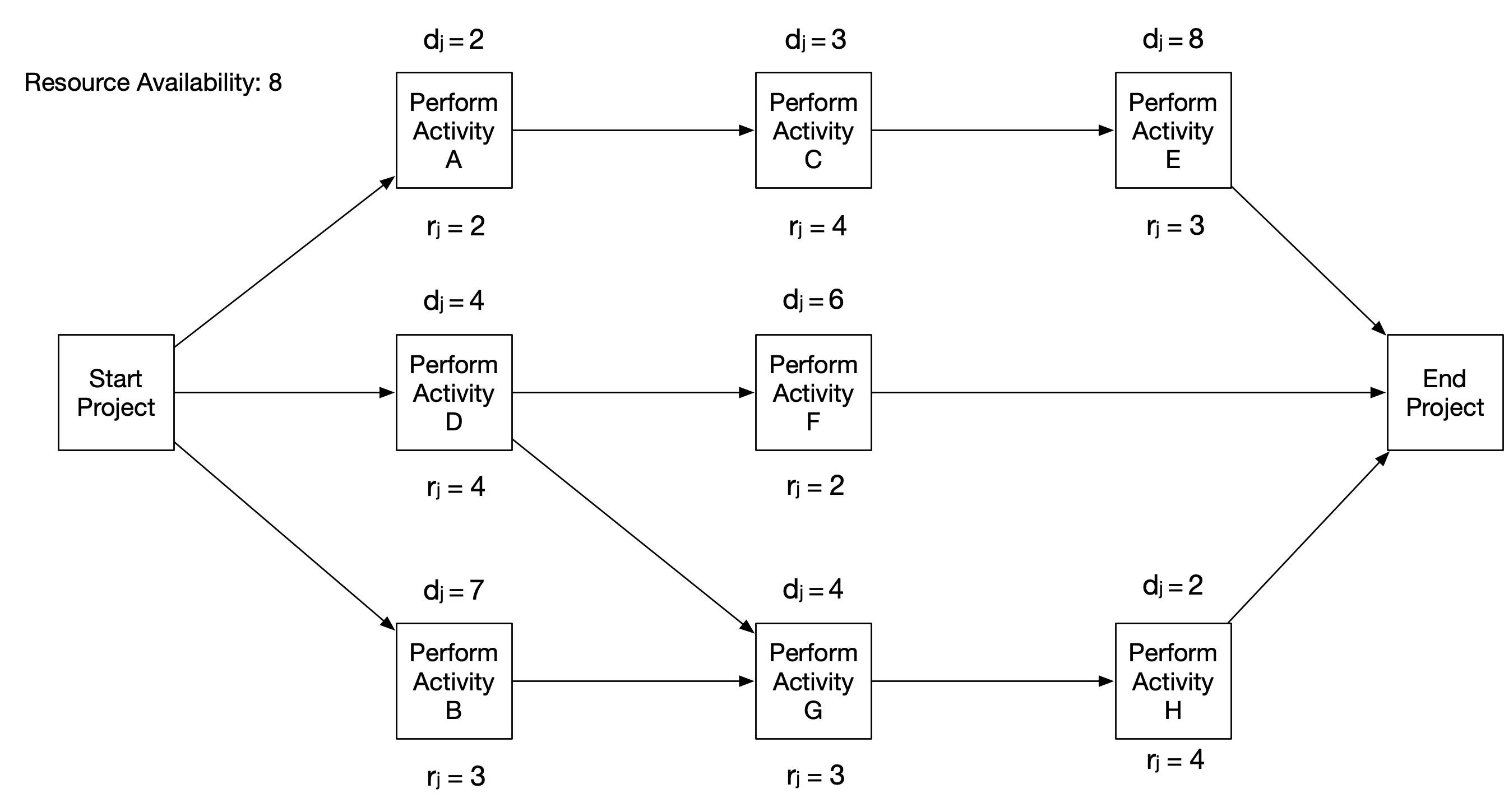}{An Example AoN Project Network for the $\emptyset,m,1|cpm|C_{max}$ variant of the RCPSP problem. \cite{demeulemeester:2002:00}}{fig:SampleProjectNetwork}{0.9}

To facilitate the remainder of the discussion, the project \cite{demeulemeester:2002:00} shown in Fig.~\ref{fig:SampleProjectNetwork} is chosen as a motivational example. The presented activity-on-nodes (AoN) network $G=\{{\cal V},{\cal A}\}$ depicts each activity ${\cal v}_{\cal i}$ as a box, containing the activity identifier, with its duration $d_{\cal i}$ indicated above the box and its resource requirements ${\cal r}_{il}$ specified below.  Activities A and J represent dummy activities, denoting the project's start and finish, respectively, each with a duration of zero and zero resource requirements.  Two variants of the Resource-Constrained Project Scheduling Problem (RCPSP) are treated in this work:  a.) with renewable resources and b.) with non-renewable resources.    The RCPSP with renewable resources is denoted as $\emptyset,m,1|cpm|C_{\max}$, where $\alpha=\emptyset,m,1$.  $\emptyset$ means there are no structural resources, $m=|S_{\cal R}|$ indicates the number of resources, and $1$ denotes that they are renewable.  Meanwhile, $\beta=cpm$ indicates a strict finish-to-start precedence relation between activities, and $\gamma=C_{max}$ means the objective function is to minimize the project's makespan.  In contrast, the RCPSPS with strictly non-renewable resources is denoted as $\emptyset,m,T|cpm|C_{max}$ where $\alpha=\emptyset,m,T$. $\emptyset$ means there are no structural resources, $m=|S_{\cal R}|$ indicates the number of resources, and $T$ denotes that they are non-renewable and their availability specified for the entire project horizon, $\beta=cpm$ means there is strict finish to start precedency relation between activities, and $\gamma=C_{max}$ means the objective function is to minimize the project makespan.  In both problems, each activity ${\cal v}_{\cal i} \in {\cal V}$ has its associated duration $d_{\cal i} \in D$, and a set of precedence arcs ${\cal A}_{i}$.  The AoN network is assumed to be well-formed so that when all activities are complete, the project has proceeded from its dummy start activity to its dummy finish activity.   Additionally, each combination of activity ${\cal v}_{\cal i}$ and resource $s_{\cal l}$ has its associated resource requirement ${\cal r}_{\cal il} \in {\cal R}$. Additionally, each resource $s_{\cal l}$ has its availability, which takes a value of $C_{\cal lRr}$ when the resource is renewable and $C_{\cal lRn}$ when it is non-renewable.  Both problems have the same objective function in Eq. \ref{Eq:RCPSP_Objective} and the constraints in Eqs. \ref{Eq:RCPSP_constraint1}--\ref{Eq:RCPSP_constraint2}. Additionally, Eq. \ref{Eq:RCPSP_constraint3} imposes a constraint on the project's renewable resources when concurrent ongoing activities are being executed.  In contrast, when the renewable resources of the RCPSP above are replaced entirely with non-renewable ones, Eq. \ref{Eq:RCPSP_constraint3} is replaced with \ref{Eq:RCPSP_constraint4} so that the availability of the non-renewable resources is imposed on the project horizon \cite{NUDTASOMBOON:1997:00,coelho:2011:00}.  Note that it is identical in form to Eq. \ref{Eq:RCPSP_constraint3} except for the summation over the time index k.  This formulation was first introduced by Pritsker et al.\cite{pritsker:1969:00} and later extended to incorporate several real-world scenarios as found in \cite{kone:2011:00,coelho:2011:00}. 

\nomenclature[I]{${\cal i}$}{Activity index ${\cal i} \in \{1, \dots, n\}$}
\nomenclature[I]{${\cal l}$}{Resource index ${\cal l} \in \{1, \dots, \lvert S_{\cal R}\rvert\}$}
\nomenclature[D]{$k$}{Time step index $k \in \{1, \ldots, K$\}}
\nomenclature[P]{$d_{\cal i}$}{Duration of activity ${\cal i}$}
\nomenclature[P]{${\cal r}_{\cal il}$}{Resource requirement of activity ${\cal i}$ from resource ${\cal l}$ }
\nomenclature[D]{${\cal x_{\cal ik}}$}{Equal to $1$ if activity ${\cal i}$ starts at timestep  $k$}

\begin{subequations}
\begin{alignat}{3}
\min \; \sum_{k = 1}^{K} & k \cdot {\cal x}_{nk}  \label{Eq:RCPSP_Objective}\\
\text{s.t.}\quad \sum_{k = 1}^{K} & {\cal x}_{{\cal i}k} &&=1 && \quad \forall {\cal i} = \{1, \dots, n\} \label{Eq:RCPSP_constraint1} \\
\sum_{k = 1}^{K} & k \cdot {\cal x}_{jk} - \left(\sum_{k = 1}^{K} k \cdot {\cal x}_{ik} + d_{\cal i} \right) &&\geq 0 && \quad \forall ({\cal i}, {\cal j}) \in {\cal A} \label{Eq:RCPSP_constraint2} \\
\mbox{(\emph{Renewable Resources})}\quad \sum_{{\cal i} = 1}^n  \sum_{\kappa=k-d_i}^{\kappa=k} & {\cal r}_{{\cal il}} \, {\cal x}_{{\cal i}\kappa} && \leq C_{\cal lRr} &&\quad \forall {\cal l} = \{1, \dots, |S_{\cal R}|\},\forall k \in \{1, \dots, K\}  \label{Eq:RCPSP_constraint3}\\\nonumber 
& \qquad\qquad\qquad\quad OR && \\
\mbox{(\emph{Non-Renewable Resources})}\quad \sum_{i = 1}^n\sum_{k = 1}^\kappa & {\cal r}_{{\cal i}l} \, {\cal x}_{{\cal i}k} && \leq C_{\cal lRn}  && \quad \forall {\cal l} = \{1, \dots, |S_{\cal R}|\},\forall \kappa \in \{1, \dots, K\}\label{Eq:RCPSP_constraint4}  
\end{alignat}
\end{subequations}
where: 
\begin{itemize}
\item $x_{{\cal i}k}$: binary decision variable equal to 1 if activity $\cal i$ starts at time $k$, and 0 otherwise.
\item Eq.~\ref{Eq:RCPSP_Objective}: minimizes the project makespan by assigning a start time to the terminal activity $n$.
\item Eq.~\ref{Eq:RCPSP_constraint1}: ensures that each activity starts exactly once within the project's time horizon.  In many cases, for computational efficiency, the project time horizon is replaced with a time interval $[EST_{\cal i},LST_{\cal i}]$ for each activity ${\cal i}$ where the earliest starting time $EST_{\cal i}$ and the latest starting time $LST_{\cal i}$ have been precalculated.  
\item Eq.~\ref{Eq:RCPSP_constraint2}: enforces precedence constraints by requiring successor activities to start only after their predecessor activities finish.
\item Eq.~\ref{Eq:RCPSP_constraint3}: enforces resource feasibility by ensuring that, at any time $k$, the total resource demand across concurrent ongoing activities does not exceed the available capacity $C_{\cal lRr}$.  Note that the summation over the time interval $\kappa = \{k-di,k\}$ measures whether an activity is ongoing or not by including any nonzero start times over that interval.  
\item Eq.~\ref{Eq:RCPSP_constraint4}: enforces resource feasibility by ensuring that for the entire project duration, the total resource demand across all activities does not exceed the available capacity $C_{\cal lRn}$.  
\end{itemize}

\vspace{-0.1in}
\subsection{Model-Based Systems Engineering}\label{Sec:MBSEIntro}
In the meantime, Model-Based Systems Engineering (MBSE) has emerged as a transformative approach for graphically modeling large, complex systems throughout their lifecycle \cite{Bajaj:2011:00}.  More interestingly, MBSE has been able to represent and analyze enterprise architectures (EA) such as those required to manage complex mega-projects \cite{mordecai:2022:00, sitton:2018:00}.  MBSE employs the Systems Modeling Language (SysML) to capture requirements, structure, behavior, and their interdependencies across multiple layers of an enterprise \cite{Kaslow2016CubeSatMB, schindel12015accelerating}. By enhancing traceability, improving stakeholder communication, and enabling real-time analysis, MBSE renders enterprise-scale complexity explicit and analyzable.  Such system-level thinking and analytical capabilities are indispensable in mega-projects that span multiple organizations, contracts, and life-cycle phases \cite{Madni2018ModelbasedSE, Dean2012ModelBasedSE, sitton:2018:00, dezhboro:2024:00}.  Within this MBSE-centric view, EA provides the strategic intent and governance objectives that the models must satisfy \cite{papke2020implementing}. EA is a systematic and holistic approach to designing and managing an organization's information systems and related capabilities, ensuring alignment with enterprise strategy \cite{busch:2025:00}. In practice, EA frames architecture evaluations--e.g., whether the enterprise can coordinate multi-project workflows at scale, absorb scope changes, and synchronize shared resources and dependencies across contractors--within decision contexts such as project planning, schedule risk management, and resource optimization \cite{gellweiler:2020:00, haren:2011:00}. MBSE operationalizes these EA concerns by mapping strategic objectives to modeled structures, behaviors, and performance measures that can be queried, simulated, and optimized for making decisions in mega-projects \cite{zhang2025mbse}. Studies in the literature reinforce this perspective by connecting MBSE, enterprise architecture, and project execution. Sitton and Reich \cite{sitton:2018:00} demonstrate how MBSE can be used to formalize and verify enterprise architecture frameworks. Their work demonstrates that MBSE enhances the completeness, traceability, integrity, and interoperability of enterprise processes, as evidenced by several case studies. In addition, Atenico et al. \cite{atencio:2022:00} provide a review of enterprise architecture as a governance instrument, highlighting its ability to have structural layers (e.g., strategy, business, application, technology) and its role in aligning projects with enterprise strategy. These works emphasize that EA has a strong influence on project outcomes by coordinating workflows, integrating project governance mechanisms, and ensuring that projects are aligned with enterprise-wide objectives. Together, these findings underscore that because EA profoundly influences mega-project execution, MBSE offers a promising approach to model, analyze, and operationalize enterprise architectures for mega-project management \cite{zhang2025mbse,sitton:2018:00, atencio:2022:00}. In summary, EA sets the \emph{why} and \emph{where} of strategic alignment and governance, while MBSE supplies the \emph{what} and \emph{how} of formal artifacts and analyses that make enterprise architectures executable for megaproject planning, risk management, and resource coordination.

\begin{figure}[htbp]
\centering
\includegraphics[width=\textwidth]{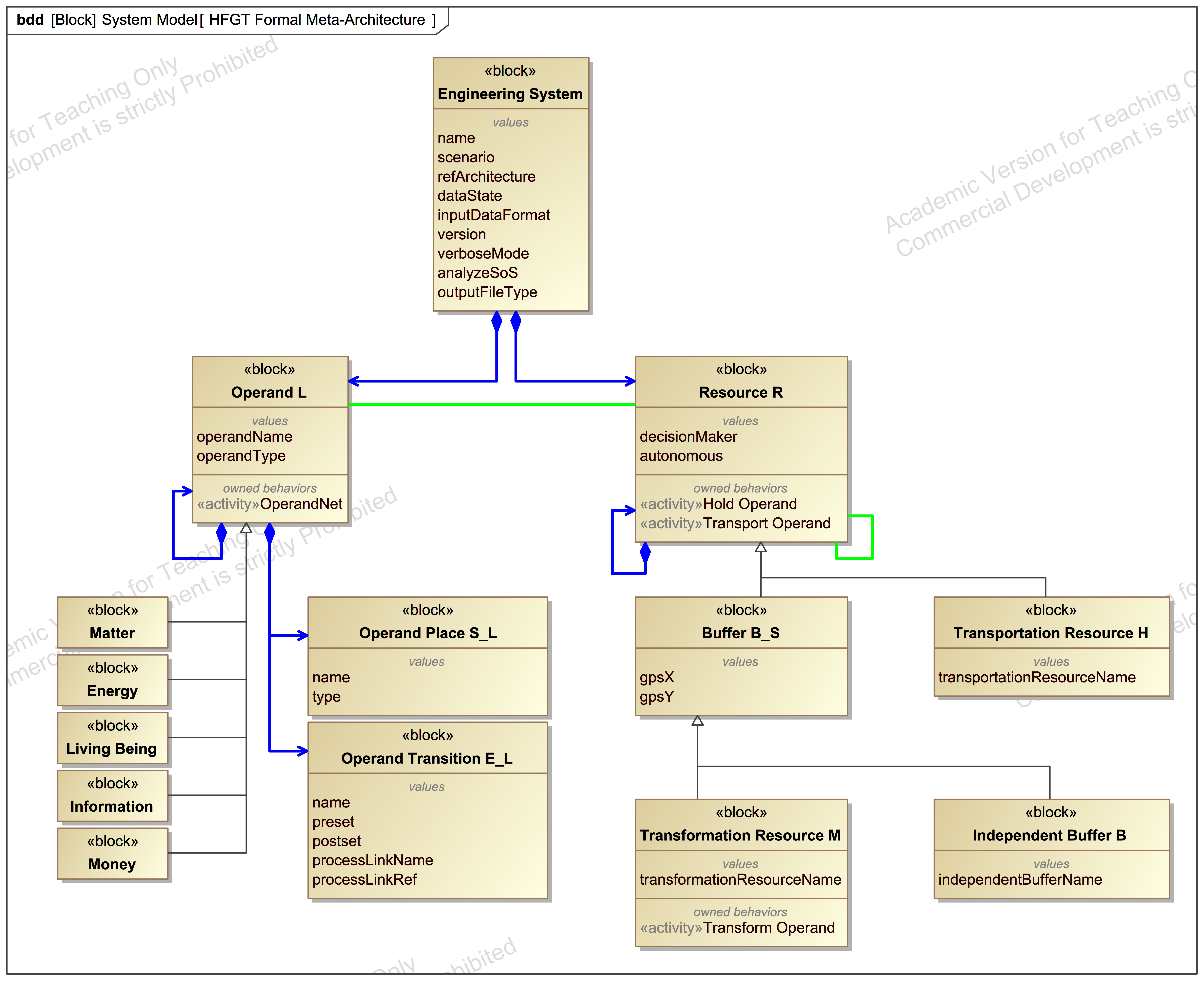}
\caption{A SysML Block Definition Diagram of the System Form of the Engineering System Meta-Architecture\cite{Schoonenberg:2019:ISC-BK04}.}
\label{Fig:LFESMetaArchitecture}
\end{figure}

While a complete introduction to Model-Based Systems Engineering (MBSE) and SysML is beyond the scope of this paper, the essential elements of two key SysML diagrams are introduced to understand how to represent a system's architecture\cite{Crawley2015SystemAS}. The diagram in Fig. \ref{Fig:LFESMetaArchitecture}, called the Block Definition Diagram (BDD), captures the form of the engineering system (or enterprise) in terms of its constituent elements and relationships. In the case of Fig. \ref{Fig:LFESMetaArchitecture}, the BDD shows a \emph{meta-architecture} without drawing all of the \emph{instances} of each block and with each block being stated in a language or vocabulary that is independent of any specific application domain\cite{Schoonenberg:2019:ISC-BK04}.  In the context of this paper, a BDD contains (at a minimum): 

\begin{itemize}
\item \textbf{Blocks:} that represent system elements, components, or subsystems.
\item \textbf{Attributes:} that represent characteristics or properties of each block.
\item \textbf{Operations:} that represent functions, activities, or processes that a block can perform.  Each combination of an operation in a block describes a capability.  
\end{itemize}

\begin{figure}
    \centering
    \includegraphics[width=1\linewidth]{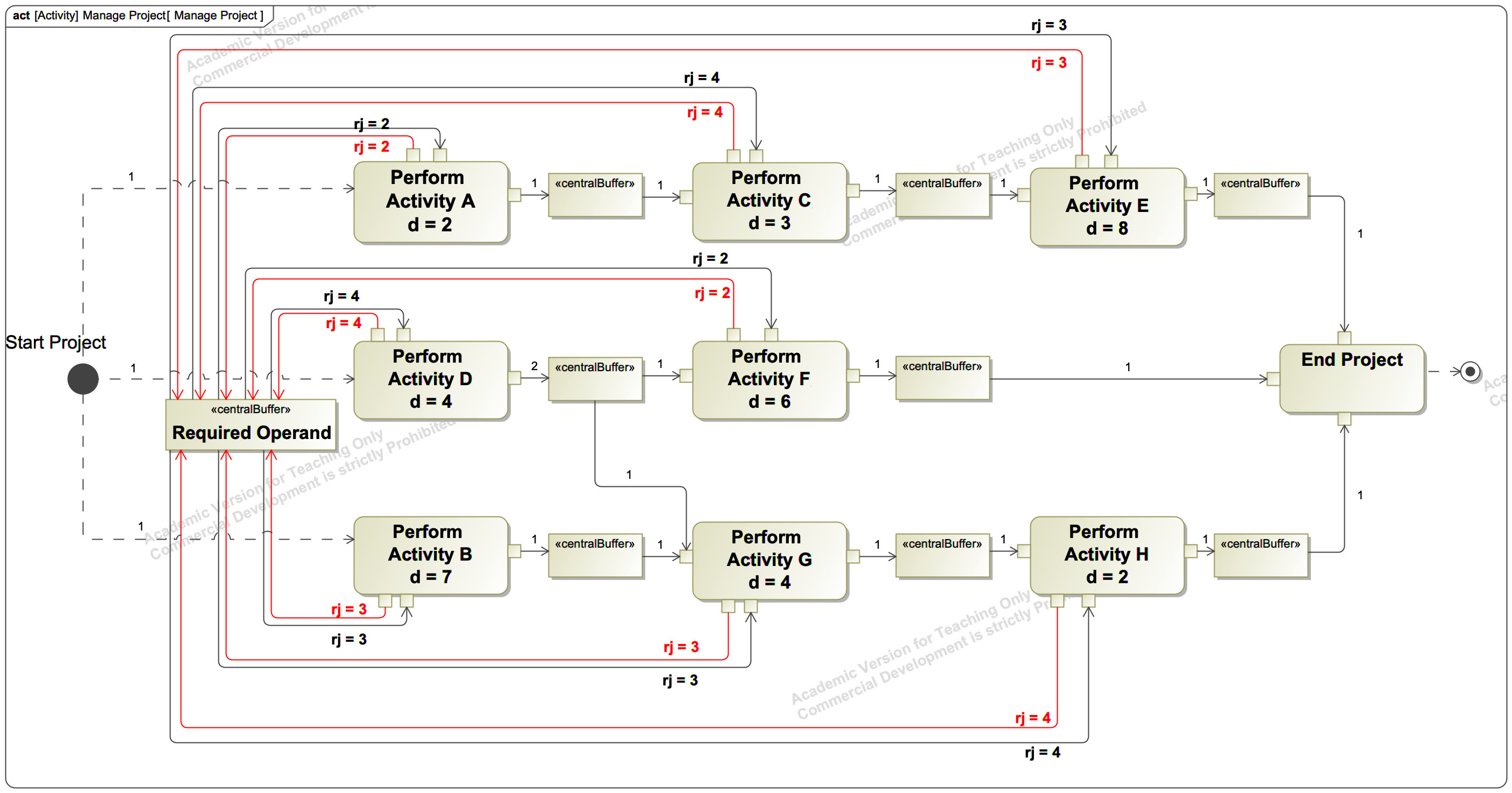}
    \caption{A  SysML Activity Diagram corresponding to the AoN project network shown in Fig. \ref{fig:SampleProjectNetwork}.}
    \label{fig:SampleProjectACT}
\end{figure}

The second diagram, called the Activity Diagram (ACT), is used to model a system's function in terms of its constituent activities. An example ACT is shown in Fig.~\ref{fig:SampleProjectACT} (and is elaborated further in Sec.~\ref {Sec:Reconciliation}). In the context of this paper, an ACT contains (at a minimum):

\begin{itemize}
\item \textbf{Actions (or Activities)}: Represent specific tasks, operations, or processes within the system.  Note that the operations found in a BDD are often represented as actions in an ACT.  Similarly, there is no distinction between a process in a process flow diagram and an action in an ACT.  
\item \textbf{Arcs}: Arrows indicate the flow from one action to another, showing the sequential or parallel execution of processes.
\item \textbf{Central Buffers:} Buffers that instantiate a pool of resources in this diagram. The object flows demonstrate the flow of required operands to activities. Object flows in red color depicts the release of renewable resources after being seized by each activity. 
\end{itemize}
Together, the activity diagram illustrates how the system's functions are executed to reveal its behavior.  

\vspace{-0.1in}
\subsection{Hetero-functional Graph Theory Definitions}\label{Sec:HFGTIntro}
While MBSE, and more specifically SysML, can graphically model large, complex systems, it does not have in-built functionality for conducting quantitative analysis.   Fortunately, hetero-functional graph theory (HFGT) provides an analytical means of translating graphical SysML models into mathematical models \cite{Farid:2025:ISC-JR06}.  As shown in Fig. \ref{Fig:LFESMetaArchitecture}, this translation requires the HFGT meta-architecture stated in SysML.  

The HFGT meta-architecture introduces several meta-elements whose definitions are formally introduced here:  
\begin{defn}[System Operand \cite{SE-Handbook-Working-Group:2015:00}]\label{Defn:D1 System Operand}
An asset or object $l_i \in L$ that is operated on or consumed during the execution of a process.
\end{defn}

\nomenclature[Y]{$l_i\in L$}{Operands}
\nomenclature[Z]{$i$}{Operand Index}

\begin{defn}[System Process\cite{Hoyle:1998:00,SE-Handbook-Working-Group:2015:00}]\label{Defn:D2 System Process}
An activity $p_w \in P$ that transforms or transports a predefined set of input operands into a predefined set of outputs. 
\end{defn}

\nomenclature[Y]{$p_w \in P$}{Processes}
\nomenclature[Z]{$w$}{Process Index}

\begin{defn}[System Resource \cite{SE-Handbook-Working-Group:2015:00}]\label{Defn:SystemResource}
An asset or object $r_v \in R$ that facilitates the execution of a process.  
\end{defn}

\nomenclature[Y]{$r_v\in R$}{Resources}
\nomenclature[Z]{$v$}{Resource Index}

\noindent Importantly, these meta-elements are organized around the universal structure of human language.  Namely, system resources $R$ serve as subjects, system processes $P$ serve as predicates, and operands $L$ serve as objects within the predicates.  The system resources $R=M \cup B \cup H$ are classified into transformation resources $M$, independent buffers $B$, and transportation resources $H$.  

\begin{rem}
It is critical to recognize that HFGT and the RCPSP do not define the word ``resource" equivalently.  In reality, what the RCPSP calls a resource, HFGT calls an operand, and what HFGT calls a resource, the machine scheduling taxonomy calls a ``structural resource" as defined by the $\alpha_1$ property.  Because the RCPSP taxonomy forces $\alpha_1=\emptyset$ and does not address structural resources, this work adopts the HFGT definitions of operand and resource instead of those of the RCPSP.  This maintains the distinction between operands and resources.  Consequently, what the RCPSP calls ``resource requirements" are called ``operand requirements" for the remainder of the paper.  Furthermore, what are called ``renewable and non-renewable resources" are called ``renewable and non-renewable operands" for the remainder of the paper.  For simplicity, the widely accepted acronym RCPSP is retained in the paper, although strictly speaking, it refers to the ``operand-constrained project scheduling problem."  As Sec \ref{Sec:Reconciliation} elaborates, this choice of definitions also facilitates the concordance between the RCPSP and HFGT.
\end{rem}

Returning to HFGT, the set of ``buffers" $ B_S=M\cup B$ is introduced to support the discussion.  
\begin{defn}[Buffer\cite{Schoonenberg:2019:ISC-BK04,Farid:2022:ISC-J49}]\label{Defn:D4 Buffer}
A resource $r_v \in R$ is a buffer $b_s \in B_S$ if it is capable of storing or transforming one or more operands at a unique location in space.  
\end{defn}

\nomenclature[Y]{$\epsilon_{\psi} \in {\cal E}_S$}{Capabilities, Engineering System Net Transitions}
\nomenclature[Z]{$\psi$}{Capability Index}

Equally important, the system processes $P = P_\mu \cup P_{\bar{\eta}}
$ are classified into transformation processes $P_\mu$ and refined transportation processes $P_\eta
$.  The latter arises from the simultaneous execution of one transportation process and one holding process.  Finally, hetero-functional graph theory emphasizes that resources are capable of one or more system processes to produce a set of ``capabilities"\cite{Schoonenberg:2019:ISC-BK04}.
\begin{defn}[Capability\cite{Schoonenberg:2019:ISC-BK04,Farid:2022:ISC-J49,Farid:2016:ISC-BC06}]\label{Defn:Capability}
An action $e_{wv} \in {\cal E}_S$ (in the SysML sense) defined by a system process $p_w \in P$ being executed by a resource $r_v \in R$.  It constitutes a subject + verb + operand sentence of the form: ``Resource $r_v$ does process $p_w
$".  
\end{defn}
\noindent The highly generic and abstract nature of these definitions has allowed HFGT to be applied to numerous application domains, including electric power, potable water, wastewater, natural gas, oil, coal, multi-modal transportation, mass-customized production, and personalized healthcare delivery systems.  For a more in-depth description of HFGT, readers are directed to past works\cite{Schoonenberg:2019:ISC-BK04,Farid:2022:ISC-J49,Farid:2016:ISC-BC06, Ghorbanichemazkati:2025:ISC-JR02}.

Returning to Fig. \ref{Fig:LFESMetaArchitecture}, the engineering system meta-architecture stated in SysML must be instantiated and ultimately transformed into the associated Petri net model. To that end, the positive and negative hetero-functional incidence tensors (HFIT) are introduced to describe the flow of operands through buffers and capabilities.  
\begin{defn}[The Negative 3$^{rd}$ Order Hetero-functional Incidence Tensor (HFIT) $\widetilde{\cal M}_\rho^-$\cite{Farid:2022:ISC-J49}]\label{Defn:D6 HFIT -ve}
The negative hetero-functional incidence tensor $\widetilde{\cal M_\rho}^- \in \{0,1\}^{|L|\times |B_S| \times |{\cal E}_S|}$  is a third-order tensor whose element $\widetilde{\cal M}_\rho^{-}(i,y,\psi)=1$ when the system capability ${\epsilon}_\psi \in {\cal E}_S$ pulls operand $l_i \in L$ from buffer $b_{s_y} \in B_S$.
\end{defn} 
\begin{defn}[The Positive  3$^{rd}$ Order Hetero-functional Incidence Tensor (HFIT)$\widetilde{\cal M}_\rho^+$\cite{Farid:2022:ISC-J49}]\label{Defn:D7 HFIT +ve}
 The positive hetero-functional incidence tensor $\widetilde{\cal M}_\rho^+ \in \{0,1\}^{|L|\times |B_S| \times |{\cal E}_S|}$  is a third-order tensor whose element $\widetilde{\cal M}_\rho^{+}(i,y,\psi)=1$ when the system capability ${\epsilon}_\psi \in {\cal E}_S$ injects operand $l_i \in L$ into buffer $b_{s_y} \in B_S$.
\end{defn}

\nomenclature[V]{$\widetilde{\cal M_\rho}^-$}{Negative 3$^{rd}$ Order Hetero-functional Incidence Tensor}
\nomenclature[V]{$\widetilde{\cal M_\rho}^+$}{Positive 3$^{rd}$ Order Hetero-functional Incidence Tensor}
\nomenclature[V]{$\widetilde{\cal M_\rho}$}{3$^{rd}$ Order Hetero-functional Incidence Tensor}
\nomenclature[V]{$\widetilde{M_\rho}^-$}{Negative 2$^{nd}$ Order Hetero-functional Incidence Tensor}
\nomenclature[V]{$\widetilde{M_\rho}^+$}{Positive 2$^{nd}$ Order Hetero-functional Incidence Tensor}
\nomenclature[V]{$\widetilde{M_\rho}$}{2$^{nd}$ Order Hetero-functional Incidence Tensor}

\noindent These incidence tensors are straightforwardly ``matricized" to form the 2$^{nd}$ Order Hetero-functional Incidence Matrix $M = M^+ - M^-$ with dimensions $|L||B_S|\times |{\cal E}|$. Consequently, the supply, demand, transportation, storage, transformation, assembly, and disassembly of multiple operands in distinct locations over time can be described by an Engineering System Net and its associated State Transition Function\cite{Schoonenberg:2022:ISC-J48}.
\begin{defn}[Engineering System Net\cite{Schoonenberg:2022:ISC-J48}]\label{Defn:ESN}
An elementary Petri net ${\cal N} = \{S, {\cal E}_S, \textbf{M}, W, Q\}$, where
\begin{itemize}
\item $S$ is the set of places with size: $|L||B_S|$,
\item ${\cal E}_S$ is the set of transitions with size: $|{\cal E}|$,
\item $\textbf{M}$ is the set of arcs, with the associated incidence matrices: $M = M^+ - M^-$,
\item $W$ is the set of weights on the arcs, as captured in the incidence matrices,
\item $Q=[Q_B; Q_E]$ is the marking vector for both the set of places and the set of transitions. 
\end{itemize}
\end{defn}
\nomenclature[Y]{${\cal N}_S$}{Engineering System Net}
\nomenclature[Y]{$S$}{Engineering System Net Places}
\nomenclature[Y]{$\textbf{M}$}{Engineering System Net Arcs}
\nomenclature[V]{$W$}{Engineering System Net Arc Weights}
\nomenclature[W]{$Q$}{Engineering System Net Marking Vector}
\nomenclature[W]{$Q_B$}{Engineering System Net Place Marking Vector}
\nomenclature[W]{$Q_{\cal E}$}{Engineering System Net Transition Marking Vector}
\nomenclature[V]{$\mathbb{M}$}{Incidence Matrix}
\nomenclature[V]{$\mathbb{M}^+$}{Positive Incidence Matrix}
\nomenclature[V]{$\mathbb{M}^-$}{Negative Incidence Matrix}

\begin{defn}[Engineering System Net State Transition Function\cite{Schoonenberg:2022:ISC-J48}]\label{Defn:ESN-STF}
The  state transition function of the engineering system net $\Phi()$ is:
\begin{equation}\label{Eq:PhiCPN}
Q[k+1]=\Phi(Q[k],U^-[k], U^+[k]) \quad \forall k \in \{1, \dots, K\}
\end{equation}
where $k$ is the discrete time index, $K$ is the simulation horizon, $Q=[Q_{B}; Q_{\cal E}]$, $Q_B$ has size $|L||B_S| \times 1$, $Q_{\cal E}$ has size $|{\cal E}_S|\times 1$, the input firing vector $U^-[k]$ has size $|{\cal E}_S|\times 1$, and the output firing vector $U^+[k]$ has size $|{\cal E}_S|\times 1$.  
\begin{align}\label{Eq:ESNSTF1}
Q_{B}[k+1]&=Q_{B}[k]+{M}^+U^+[k]\Delta T-{M}^-U^-[k]\Delta T \\ \label{Eq:ESNSTF2}
Q_{\cal E}[k+1]&=Q_{\cal E}[k]-U^+[k]\Delta T +U^-[k]\Delta T
\end{align}
where $\Delta T$ is the duration of the simulation time step.  
\end{defn}

\nomenclature[Y]{$\Phi_S()$}{Engineering System Net State Transition Function}
\nomenclature[W]{$U_S^-$}{Engineering System Net Input Firing Vector}
\nomenclature[W]{$U_S^+$}{Engineering System Net Output Firing Vector}
\nomenclature[W]{$U_S$}{Engineering System Net Firing Vector}
\nomenclature[V]{$\Delta T$}{Duration of Simulation Time Step}

In addition to the engineering system net, in HFGT, each operand can have its own state and evolution.  This behavior is described by an Operand Net and its associated State Transition Function for each operand.  
\begin{defn}[Operand Net\cite{Farid:2008:IEM-J04,Schoonenberg:2019:ISC-BK04,Khayal:2017:ISC-J35,Schoonenberg:2017:IEM-J34}]\label{Defn:OperandNet} Given operand $l_i$, an elementary Petri net ${\cal N}_{l_i}= \{S_{l_i}, {\cal E}_{l_i}, \textbf{M}_{l_i}, W_{l_i}, Q_{l_i}\}$ where 
\begin{itemize}
\item $S_{l_i}$ is the set of places describing the operand's state.  
\item ${\cal E}_{l_i}$ is the set of transitions describing the evolution of the operand's state.
\item $\textbf{M}_{l_i} \subseteq (S_{l_i} \times {\cal E}_{l_i}) \cup ({\cal E}_{l_i} \times S_{l_i})$ is the set of arcs, with the associated incidence matrices: $M_{l_i} = M^+_{l_i} - M^-_{l_i} \quad \forall l_i \in L$.  
\item $W_{l_i} : \textbf{M}_{l_i}$ is the set of weights on the arcs, as captured in the incidence matrices $M^+_{l_i},M^-_{l_i} \quad \forall l_i \in L$.  
\item $Q_{l_i}= [Q_{Sl_i}; Q_{{\cal E}l_i}]$ is the marking vector for both the set of places and the set of transitions. 
\end{itemize}
\end{defn}

\nomenclature[Y]{${\cal N}_{l_i}$}{$l_i^{th}$ Operand Net}
\nomenclature[Y]{$S_{l_i}$}{$l_i^{th}$ Operand Net Places}
\nomenclature[Y]{$\epsilon_{\chi l_i} $}{$l_i^{th}$ Operand Net Transitions $\in {\cal E}_{l_i}$}
\nomenclature[Y]{$\textbf{M}_{l_i}$}{$l_i^{th}$ Operand Net Arcs}
\nomenclature[V]{$M_{l_i}$}{$l_i^{th}$ Operand Net Incidence Matrix}
\nomenclature[V]{$M_{l_i}^+$}{$l_i^{th}$ Operand Net Positive Incidence Matrix}
\nomenclature[V]{$M_{l_i}^-$}{$l_i^{th}$ Operand Net Negative Incidence Matrix}
\nomenclature[V]{$W_{l_i}$}{$l_i^{th}$ Operand Net Arc Weights}
\nomenclature[W]{$Q_{l_i}$}{$l_i^{th}$ Operand Net Marking Vector}
\nomenclature[W]{$Q_{Sl_i}$}{$l_i^{th}$ Operand Net Place Marking Vector}
\nomenclature[W]{$Q_{{\cal E}l_i}$}{$l_i^{th}$ Operand Net Transition Marking Vector}

\begin{defn}[Operand Net State Transition Function\cite{Farid:2008:IEM-J04,Schoonenberg:2019:ISC-BK04,Khayal:2017:ISC-J35,Schoonenberg:2017:IEM-J34}]\label{Defn:OperandNet-STF}
The  state transition function of each operand net $\Phi_{l_i}()$ is:
\begin{equation}\label{Eq:PhiSPN}
Q_{l_i}[k+1]=\Phi_{l_i}(Q_{l_i}[k],U_{l_i}^-[k], U_{l_i}^+[k]) \quad \forall k \in \{1, \dots, K\} \quad i \in \{1, \dots, L\}
\end{equation}
where $Q_{l_i}=[Q_{Sl_i}; Q_{{\cal E} l_i}]$, $Q_{Sl_i}$ has size $|S_{l_i}| \times 1$, $Q_{{\cal E} l_i}$ has size $|{\cal E}_{l_i}| \times 1$, the input firing vector $U_{l_i}^-[k]$ has size $|{\cal E}_{l_i}|\times 1$, and the output firing vector $U^+[k]$ has size $|{\cal E}_{l_i}|\times 1$.  

\begin{align}\label{EQ:ON-STF}
Q_{Sl_i}[k+1]&=Q_{Sl_i}[k]+{M_{l_i}}^+U_{l_i}^+[k]\Delta T - {M_{l_i}}^-U_{l_i}^-[k]\Delta T \\ \label{CH6:CH eq:Q_E:HFNMCFprogram}
Q_{{\cal E} l_i}[k+1]&=Q_{{\cal E} l_i}[k]-U_{l_i}^+[k]\Delta T +U_{l_i}^-[k]\Delta T
\end{align}
\end{defn}

\nomenclature[Y]{$\Phi_{l_i}()$}{$l_i^{th}$ Operand Net State Transition Function}
\nomenclature[W]{$U_{l_i}^-$}{$l_i^{th}$ Operand Net Input Firing Vector}
\nomenclature[W]{$U_{l_i}^+$}{$l_i^{th}$ Operand Net Output Firing Vector}
\nomenclature[W]{$U_{l_i}$}{$l_i^{th}$ Operand Net Firing Vector}

\nomenclature[Z]{$\chi$}{Operand Net Transition Index}

Ultimately, the reconciliation of the RCPSP, model-based systems engineering, and hetero-functional graph theory in Sec. \ref{Sec:Reconciliation} relies on comparing the mathematics of RCPSP with the operand net state transition function found in Defn. \ref{Defn:OperandNet-STF}.  The engineering system net state transition function (Defn. \ref{Defn:ESN-STF}) is not part of the RCPSP reconciliation because $\alpha_1=0$, but is required in the management of mega-projects as Sec. \ref{Sec:Discussion} discusses.

\vspace{-0.1in}
\subsection{Hetero-functional Network Minimum Cost Flow Optimization}\label{Sec:HFNMCF_Intro}
HFGT describes the behavior of an engineering system using the Hetero-Functional Network Minimum Cost Flow (HFNMCF) problem\cite{Schoonenberg:2022:ISC-J48}. It optimizes the time-dependent flow and storage of multiple operands (or commodities) between buffers, allows for their transformation from one operand to another, and tracks the state of these operands.  In this regard, it is a very flexible optimization problem that applies to a wide variety of complex engineering systems.  For the purposes of this paper, the HFNMCF is a type of discrete-time-dependent, time-invariant, convex optimization program\cite{Schoonenberg:2022:ISC-J48}.

\newpage 
\begin{subequations}
\begin{align}\label{Eq:HFNMCF_ObjFunc}
Z = \sum_{k=1}^{K-1} x^T[k] F_{QP}[k] x[k] + f^T_{QP}[k] x[k] 
\end{align}
\vspace{-0.2in}
\begin{align}\label{Eq:HFNMCF_ESN-STF1}
\text{s.t. } -Q_{B}[k+1]+Q_{B}[k]+{M}^+U^+[k]\Delta T - {M}^-U^-[k]\Delta T=&0 && \!\!\!\!\!\!\!\!\!\!\!\!\!\!\!\!\!\!\!\!\!\!\!\!\!\!\!\!\!\!\!\!\!\!\!\!\!\!\!\!\!\forall k \in \{1, \dots, K\}\\  \label{Eq:HFNMCF_ESN-STF2}
-Q_{\cal E}[k+1]+Q_{\cal E}[k]-U^+[k]\Delta T + U^-[k]\Delta T=&0 && \!\!\!\!\!\!\!\!\!\!\!\!\!\!\!\!\!\!\!\!\!\!\!\!\!\!\!\!\!\!\!\!\!\!\!\!\!\!\!\!\!\forall k \in \{1, \dots, K\}\\ \label{Eq:HFNMCF_DurationConstraint}
 - U_\psi^+[k+d_{\psi}]+ U_{\psi}^-[k] = &0 && \!\!\!\!\!\!\!\!\!\!\!\!\!\!\!\!\!\!\!\!\!\!\!\!\!\!\!\!\!\!\!\!\!\!\!\!\!\!\!\!\!\forall k\in \{1, \dots, K\}, \quad \psi \in \{1, \dots, {\cal E}_S\}\\ \label{Eq:HFNMCF_OperandNet-STF1}
 -Q_{Sl_i}[k+1]+Q_{Sl_i}[k]+{M}_{l_i}^+U_{l_i}^+[k]\Delta T - {M}_{l_i}^-U_{l_i}^-[k]\Delta T=&0 && \!\!\!\!\!\!\!\!\!\!\!\!\!\!\!\!\!\!\!\!\!\!\!\!\!\!\!\!\!\!\!\!\!\!\!\!\!\!\!\!\!\forall k \in \{1, \dots, K\}, \quad i \in \{1, \dots, |L|\}\\\label{Eq:HFNMCF_OperandNet-STF2}
-Q_{{\cal E}l_i}[k+1]+Q_{{\cal E}l_i}[k]-U_{l_i}^+[k]\Delta T + U_{l_i}^-[k]\Delta T=&0 && \!\!\!\!\!\!\!\!\!\!\!\!\!\!\!\!\!\!\!\!\!\!\!\!\!\!\!\!\!\!\!\!\!\!\!\!\!\!\!\!\!\forall k \in \{1, \dots, K\}, \quad i \in \{1, \dots, |L|\}\\ \label{Eq:HFNMCF_OperandNetDurationConstraint}
- U_{\chi l_i}^+[k+d_{\chi l_i}]+ U_{\chi l_i}^-[k] = &0 &&  \!\!\!\!\!\!\!\!\!\!\!\!\!\!\!\!\!\!\!\!\!\!\!\!\!\!\!\!\!\!\!\!\!\!\!\!\!\!\!\!\!
\forall k\in \{1, \dots, K\}, \chi\in \{1, \dots, |{\cal E}_{l_i}|\}, \: l_i \in \{1, \dots, |L|\}\\ 
\label{Eq:HFNMCF_SyncPlus}
U^+_L[k] - \widehat{\Lambda}^+ U^+[k] =&0 && \!\!\!\!\!\!\!\!\!\!\!\!\!\!\!\!\!\!\!\!\!\!\!\!\!\!\!\!\!\!\!\!\!\!\!\!\!\!\!\!\!\forall k \in \{1, \dots, K\}\\ \label{Eq:HFNMCF_SyncMinus}
U^-_L[k] - \widehat{\Lambda}^- U^-[k] =&0 && \!\!\!\!\!\!\!\!\!\!\!\!\!\!\!\!\!\!\!\!\!\!\!\!\!\!\!\!\!\!\!\!\!\!\!\!\!\!\!\!\!\forall k \in \{1, \dots, K\}\\\label{Eq:HFNMCF_ESN-Exogenous}
\begin{bmatrix}
D_{Up} & \mathbf{0} \\ \mathbf{0} & D_{Un}
\end{bmatrix} \begin{bmatrix}
U^+ \\ U^-
\end{bmatrix}[k] =& \begin{bmatrix}
C_{Up} \\ C_{Un}
\end{bmatrix}[k] && \!\!\!\!\!\!\!\!\!\!\!\!\!\!\!\!\!\!\!\!\!\!\!\!\!\!\!\!\!\!\!\!\!\!\!\!\!\!\!\!\!\forall k \in \{1, \dots, K\} \\\label{Eq:HFNMCF_OperandNet-Exogenous}
\begin{bmatrix}
E_{Lp} & \mathbf{0} \\ \mathbf{0} & E_{Ln}
\end{bmatrix} \begin{bmatrix}
U^+_{l_i} \\ U^-_{l_i}
\end{bmatrix}[k] =& \begin{bmatrix}
F_{Lpi} \\ F_{Lni}
\end{bmatrix}[k] && \!\!\!\!\!\!\!\!\!\!\!\!\!\!\!\!\!\!\!\!\!\!\!\!\!\!\!\!\!\!\!\!\!\!\!\!\!\!\!\!\!\forall k \in \{1, \dots, K\}\quad i \in \{1, \dots, |L|\} \\\label{Eq:HFNMCF_InitCond} 
\begin{bmatrix} Q_B ; Q_{\cal E} ; Q_{SL} \end{bmatrix}[1] =& \begin{bmatrix} C_{B1} ; C_{{\cal E}1} ; C_{{SL}1} \end{bmatrix} \\ \label{Eq:HFNMCF_FinalCond}
\begin{bmatrix} Q_B ; Q_{\cal E} ; Q_{SL} ; U^- ; U_L^- \end{bmatrix}[K+1] =   &\begin{bmatrix} C_{BK} ; C_{{\cal E}K} ; C_{{SL}K} ; \mathbf{0} ; \mathbf{0} \end{bmatrix}\\\label{Eq:HFNMCF_Capacity}
\underline{E}_{CP} \leq D(X) \leq& \overline{E}_{CP} \\\label{Eq:HFNMCF_DeviceModels1}
g(X,Y) =& 0 \\\label{Eq:HFNMCF_DeviceModels2}
h(Y) \leq& 0
\end{align}
\end{subequations}
where $X=\left[x[1]; \ldots; x[K]\right]$  is the vector of primary decision variables and $x[k] = \begin{bmatrix} Q_B ; Q_{\cal E} ; Q_{SL} ; Q_{{\cal E}L} ; U^- ; U^+ ; U^-_L ; U^+_L \end{bmatrix}[k]$ $\forall k \in \{1, \dots, K\}$.  The mathematical domain of these decision variables depends on the application domain (e.g., real, positive real, integer, binary values.) $Y=\left[y[1]; \ldots; y[K]\right]$ is the vector of auxiliary decision variables whose presence and nature depend on the specific problem application.  

\nomenclature[W]{$X$}{Vector of primary decision variable $x[k]$}
\nomenclature[W]{$x[k]$}{$\{ Q_B ; Q_{\cal E} ; Q_{SL} ; Q_{{\cal E}L} ; U^- ; U^+ ; U^-_L ; U^+_L \}[k]$ $\forall k \in \{1, \dots, K\}$}

\vspace{0.1in}
\subsubsection{Objective Function}
In Eq.  \ref{Eq:HFNMCF_ObjFunc}, $Z$ is a convex objective function separable in discrete time steps $k$.  The matrix $F_{QP}$ and vector $f_{QP}$ enable quadratic and linear costs to be incurred from the place and transition markings in both the engineering system net and operand nets. 
\begin{itemize}
\item $F_{QP}[k]$ is a set of positive semi-definite, diagonal, quadratic coefficient matrices.
\item $f_{QP}[k]$ are a set of linear coefficient vectors.
\end{itemize}

\vspace{0.1in}
\subsubsection{Equality Constraints}

\begin{itemize}
\item Equations \ref{Eq:HFNMCF_ESN-STF1} and \ref{Eq:HFNMCF_ESN-STF2} describe the state transition function of an engineering system net (Defn \ref{Defn:ESN} \& \ref{Defn:ESN-STF}).
\item Equation \ref{Eq:HFNMCF_DurationConstraint} is the engineering system net transition duration constraint where the end of the $\psi^{th}$ transition occurs $k_{d\psi}$ time steps after its beginning. 
\item Equations \ref{Eq:HFNMCF_OperandNet-STF1} and \ref{Eq:HFNMCF_OperandNet-STF2} describe the state transition function of each operand net ${\cal N}_{l_i}$ (Defn. \ref{Defn:OperandNet} \& \ref{Defn:OperandNet-STF}) associated with each operand $l_i \in L$.  
\item Equation \ref{Eq:HFNMCF_OperandNetDurationConstraint} is the operand net transition duration constraint where the end of the $\chi^{th}$ transition occurs $d_{\chi _{l_i}}$ time steps after its beginning. 
\item Equations \ref{Eq:HFNMCF_SyncPlus} and \ref{Eq:HFNMCF_SyncMinus} are synchronization constraints that couple the input and output firing vectors of the engineering system net to the input and output firing vectors of the operand nets, respectively. $U_L^-$ and $U_L^+$ are the vertical concatenations of the input and output firing vectors $U_{l_i}^-$ and $U_{l_i}^+$, respectively.
\begin{align}
U_L^-[k]&=\left[U^-_{l_1}; \ldots; U^-_{l_{|L|}}\right][k] \\
U_L^+[k]&=\left[U^+_{l_1}; \ldots; U^+_{l_{|L|}}\right][k]
\end{align}
\item Equations \ref{Eq:HFNMCF_ESN-Exogenous} and \ref{Eq:HFNMCF_OperandNet-Exogenous} are boundary conditions.  Eq. \ref{Eq:HFNMCF_ESN-Exogenous} is a boundary condition constraint that allows some of the engineering system net firing vectors decision variables to be set to an exogenous constant.  Eq. \ref{Eq:HFNMCF_OperandNet-Exogenous} does the same for the operand net firing vectors.  
\item Equations \ref{Eq:HFNMCF_InitCond} and \ref{Eq:HFNMCF_FinalCond} are the initial and final conditions of the engineering system net and the operand nets, where $Q_{SL}$ is the vertical concatenation of the place marking vectors of the operand nets $Q_{Sl_i}$.
\begin{align}
Q_{SL}^-[k]&=\left[Q^-_{Sl_1}; \ldots; U^-_{Sl_{|L|}}\right][k] \\
U_{SL}^+[k]&=\left[U^+_{Sl_1}; \ldots; U^+_{Sl_{|L|}}\right][k]
\end{align}
\end{itemize}

\vspace{0.1in}
\subsubsection{Inequality Constraints}
 $D_{QP}()$ and vector $E_{QP}$ in Equation \ref{Eq:HFNMCF_Capacity} place capacity constraints on the vector of primary decision variables at each time step.

\vspace{0.1in}
\subsubsection{Device Model Constraints}
g(X,Y) and h(Y) are a set of device model functions whose presence and nature depend on the specific problem application.  They can not be further elaborated until the application domain and its associated capabilities are identified.  

\vspace{-0.1in}
\section{Reconciliation of Resource-Constrained Project Scheduling Problem, Model-Based Systems Engineering, \& Hetero-functional Graph Theory}\label{Sec:Reconciliation}
While it is clear that RCPSP, model-based systems engineering, and hetero-functional graph theory can all model a diversity of complex systems,  it is also clear that they use significantly different terminology that makes it difficult to relate them to each other conceptually. To demonstrate these relationships concretely, Sec.\ref{Sec:RCPSPExample} uses the renewable and non-renewable variants of RCPSP to schedule the illustrative example introduced in Fig. \ref{fig:SampleProjectNetwork}.  Then Sec.\ref{MBSE_SecIII} represents the same project in MBSE diagrams.  Finally, Sec.\ref{sec:HFGT_SecIII} formalizes the RCPSP elements using HFGT notation and introduces an algorithm to transform the AoN network into an OperandNet for both renewable and non-renewable operand cases. The section then presents the solution of the HFNCMF problem when it is specialized to the case of the RCPSP.    
\vspace{-0.1in}
\subsection{Resource Constrained Project Scheduling}\label{Sec:RCPSPExample}
First, the renewable operand variant of the RCPSPS, as presented in Eqs.~\ref{Eq:RCPSP_Objective}--\ref{Eq:RCPSP_constraint3}, is solved and discussed.  A value of $C_{\cal lRr}=8$ is chosen.  The optimal value ${\cal x}^*$ shows when each project task begins:  
\begin{align}\label{Eq:RCPSP_OptimalXRenewable}
{\cal x}^{*} =
\left[
\begin{array}{*{18}{r|}r}
1 &   &   &   &   &   &    &    &    &    &    &    &    &    &    &    &   &  &  \\\hline
  &   & 1 &   &   &   &    &    &    &    &    &    &    &    &    &    &   &  &  \\\hline
  &   &   &   & 1 &   &    &    &    &    &    &    &    &    &    &    &   &  &  \\\hline
1 &   &   &   &   &   &    &    &    &    &    &    &    &    &    &    &   &  &  \\\hline
  &   &   &   &   &   &    & 1  &    &    &    &    &    &    &    &    &   &  &  \\\hline
  &   &   &   &   &   &    & 1  &    &    &    &    &    &    &    &    &   &  &  \\\hline
  &   &   &   &   &   &    &    &    & 1  &    &    &    &    &    &    &   &  &  \\\hline
  &   &   &   &   &   &    &    &    &    &    &    &    & 1  &    &    &   &  &  \\\hline
  &   &   &   &   &   &    &    &    &    &    &    &    &    &    & 1  &   &  &  \\
\end{array}
\right]
\end{align}
where the row index ${\cal i}$ corresponds to project tasks and the column index $k$ corresponds to time steps.  To gain further insight into the optimal schedule of the project, Table \ref{tab:RCPSP-schedule-renewable} visualizes the quantity $\sum_{\kappa=k-d_i}^{\kappa=k} {\cal r}_{{\cal il}} \, {\cal x}_{{\cal i}\kappa} \forall k \in \{1, \dots, K\} $ (introduced earlier in Eq. \ref{Eq:RCPSP_constraint3}).  It shows not only the durations during which each project activity is performed, but also the number of operands utilized when those project activities are executed.  Furthermore, Table \ref{tab:RCPSP-schedule-renewable} shows that the activity start times respect both precedence and operadnd feasibility constraints.  More specifically, the renewable capacity $C_{\cal lRr}=8$ is fully utilized in certain intervals (e.g., k=[8,\ldots,13]), while underutilization occurs elsewhere due to precedence constraints.  The resulting optimal project makespan is 15 time units -- which is interestingly two time units longer than if \ref{Eq:RCPSP_constraint3} had been relaxed, thereby quantifying the impact of renewable operand constraints on the project's duration. 

\begin{table}[h!]
\centering
\caption{Optimal Project Schedule for the Renewable Operands Variant of RCPSCP ($\emptyset,m,1|cpm|C_{max}$).}
\label{tab:RCPSP-schedule-renewable}
\begin{adjustbox}{max width=\textwidth}
\begin{tabular}{c|cccccccccccccccc}
\toprule
\textbf{Activity/Time} & 0  & 1  & 2  & 3  & 4  & 5  & 6  & 7  & 8  & 9  & 10 & 11 & 12 & 13 & 14 & 15 \\
\midrule
Perform Activity A     &    & 2  & 2  &    &    &    &    &    &    &    &    &    &    &    &    &    \\
Perform Activity B     &    &    &    & 3  & 3  & 3  & 3  & 3  & 3  & 3  &    &    &    &    &    &    \\
Perform Activity C     &    &    &    &    &    & 4  & 4  & 4  &    &    &    &    &    &    &    &    \\
Perform Activity D     &    & 4  & 4  & 4  & 4  &    &    &    &    &    &    &    &    &    &    &    \\
Perform Activity E     &    &    &    &    &    &    &    &    & 3  & 3  & 3  & 3  & 3  & 3  & 3  & 3  \\
Perform Activity F     &    &    &    &    &    &    &    &    & 2  & 2  & 2  & 2  & 2  & 2  &    &    \\
Perform Activity G     &    &    &    &    &    &    &    &    &    &    & 3  & 3  & 3  & 3  &    &    \\
Perform Activity H     &    &    &    &    &    &    &    &    &    &    &    &    &    &    & 4  & 4  \\
\midrule
\textbf{Operands Allocated Per Period} & 0  & 6  & 6  & 7  & 7  & 7  & 7  & 7  & 8  & 8  & 8  & 8  & 8  & 8  & 7  & 7 \\\hline
\textbf{Operands Allocated in Project} & 0  & 6  & 6  & 9  & 9  & 13 & 13  & 13  & 18 & 18 & 21  & 21  & 21  & 21  & 25  & 25 \\
\bottomrule
\end{tabular}
\end{adjustbox}
\end{table}

Next, the non-renewable operand variant of the RCPSPS in Eqs.~\ref{Eq:RCPSP_Objective}--\ref{Eq:RCPSP_constraint2}, \ref{Eq:RCPSP_constraint4} is solved and discussed.  A value of $C_{\cal lRn}=25$ is chosen.  The optimal value ${\cal x}^*$ shows when each project task begins:  
\begin{align}\label{Eq:RCPSP_OptimalXNonRenewable}
{\cal x}^{*} =
\left[
\begin{array}{*{18}{r|}r}
1 &   &   &   &   &   &    &    &    &    &    &    &    &    &    &    &   &  &  \\\hline
1 &   &   &   &   &   &    &    &    &    &    &    &    &    &    &    &   &  &  \\\hline
  &   &  1&   &   &   &    &    &    &    &    &    &    &    &    &    &   &  &  \\\hline
1 &   &   &   &   &   &    &    &    &    &    &    &    &    &    &    &   &  &  \\\hline
  &   &   &   &   & 1 &    &    &    &    &    &    &    &    &    &    &   &  &  \\\hline
  &   &   &   & 1 &   &    &    &    &    &    &    &    &    &    &    &   &  &  \\\hline
  &   &   &   &   &   &    & 1  &    &    &    &    &    &    &    &    &   &  &  \\\hline
  &   &   &   &   &   &    &    &    &    &    & 1  &    &    &    &    &   &  &  \\\hline
  &   &   &   &   &   &    &    &    &    &    &    &    & 1  &    &    &   &  &  \\
\end{array}
\right]
\end{align}
Again, to gain further insight into the optimal schedule of the project, Table \ref{tab:RCPSP-schedule-nonrenewable} visualizes the quantity \newline $\sum_{\kappa=k-d_i}^{\kappa=k} {\cal r}_{{\cal il}} \, {\cal x}_{{\cal i}\kappa} \forall k \in \{1, \dots, K\} $ (introduced earlier in Eq. \ref{Eq:RCPSP_constraint3}).  As in the previous example, this table shows not only the durations during which each project activity is performed, but also the number of operands utilized when those project activities are executed.  Table \ref{tab:RCPSP-schedule-nonrenewable} shows that the activity start times respect both the precedence and non-renewable operand constraints.  Interestingly, the non-renewable operand result in Table~\ref{tab:RCPSP-schedule-nonrenewable} relaxes the per-time step limit while imposing a cumulative availability constraint $C_{\cal lRn}$.  This relaxation allows multiple activities to overlap in early periods, even if their combined demand exceeds the renewable cap $C_{\cal lRr}$. Consequently, greater parallelism is achieved, reducing the makespan back to 13. 

\begin{table}[h!]
\centering
\caption{Optimal Project Schedule for the Non-renewable Operand Variant of RCPSCP $(\emptyset,m,1T|cpm|C_{max})$.}
\label{tab:RCPSP-schedule-nonrenewable}
\begin{adjustbox}{max width=\textwidth}
\begin{tabular}{c|cccccccccccccccc}
\toprule
\textbf{Activity/Time} & 0  & 1  & 2  & 3  & 4  & 5  & 6  & 7  & 8  & 9  & 10 & 11 & 12 & 13 & 14 & 15 \\
\midrule
Perform Activity A     &    & 2  & 2  &    &    &    &    &    &    &    &    &    &    &    &    &    \\
Perform Activity B     &    & 3  & 3  & 3  & 3  & 3  & 3  & 3  &    &    &    &    &    &    &    &    \\
Perform Activity C     &    &    &    & 4  & 4  & 4  &    &    &    &    &    &    &    &    &    &    \\
Perform Activity D     &    & 4  & 4  & 4  & 4  &    &    &    &    &    &    &    &    &    &    &    \\
Perform Activity E     &    &    &    &    &    &    & 3  & 3  & 3  & 3  & 3  & 3  & 3  & 3  &    &   \\
Perform Activity F     &    &    &    &    &    & 2  & 2  & 2  & 2  & 2  & 2  &    &    &    &    &    \\
Perform Activity G     &    &    &    &    &    &    &    &    & 3  & 3  & 3  & 3  &    &    &    &    \\
Perform Activity H     &    &    &    &    &    &    &    &    &    &    &    &    & 4  & 4  &    &   \\
\midrule
\textbf{Total Operands Allocated Per Period} & 0  & 9  & 9  & 11 & 11 & 9 & 8  & 8  & 8  & 8  & 8  & 6  & 7  & 7  &   &  \\\hline
\textbf{Total Operand Allocated in Project} & 0  & 9  & 9  & 13 & 13 & 15 & 18  & 18  & 21  & 21  & 21  & 21  & 25  & 25  &   &  \\
\bottomrule
\end{tabular}
\end{adjustbox}
\end{table}

\subsection{Model-Based Systems Engineering}\label{MBSE_SecIII}
To establish correspondence between the RCPSP and HFGT, it is necessary to describe an arbitrary AoN project network from an MBSE perspective.   To that end, the block definition diagram of the HFGT meta-architecture, presented in Fig.~\ref{Fig:LFESMetaArchitecture}, provides a generic and highly flexible starting point.  Again, recall that Fig.~\ref{Fig:LFESMetaArchitecture} depicts a \emph{meta-architecture}\cite{Schoonenberg:2019:ISC-BK04}, and so by definition of the application-domain-independent blocks must be related to their application-domain-specific concepts (in a project schedule).  Furthermore, Fig.~\ref{Fig:LFESMetaArchitecture}, also by definition, does not show each of the individual instances of its constituent blocks.  Turning to the top of the figure, the engineering system block represents the enterprise tasked with carrying out and managing the mega-project.   Next, because the RCPSP forces $\alpha_1=\emptyset$, the entire resource branch of Fig.~\ref{Fig:LFESMetaArchitecture} can be neglected.  Next, one instance of the operand block represents the mega-project itself because it is the enterprise's primary operand.  Then each instance of the operand transitions pertains to its associated activity in the AoN network.  Because the meta-architecture in Fig.~\ref{Fig:LFESMetaArchitecture} does not contain data attributes specific to a specific application domain, each instance of the operand transition block requires a data attribute for the activity duration $d_i$ and for each of the operand requirements ${\cal r}_{ik}$.  Next, operand places represent intermediate points of project execution before and after a task is executed and along the edges of the AoN project network.  (These are defined more specifically in Algorithm \ref{Alg:AoN-to-ACT} below.)  Finally, an additional operand place instance is used to manage the operand requirements (originally referred to as resource requirements in the RCPSP). 

\begin{algorithm}
\caption{AoN-to-ACT Construction for Renewable/Non-Renewable Operands}\label{Alg:AoN-to-ACT}
\begin{algorithmic}[1]
\Procedure{AoNtoACT}{$G=\{V,A\},D, S_{\cal R},C_{\cal R}$}
\State Create an initial node called ``Project Start".  
\State Create an action node called ``Finish Project" with a control for arc $f \in {\cal F}_c$ connected to a terminal node.    
\ForAll{$s_{\cal l} \in S_{\cal R}$}
\State Create a central buffer $s_l$ called operand requirement $s_l$ and annotate its availability $C_{\cal lRr}$ in the case of renewable operands and $C_{\cal lRn}$ in the case of non-renewable operands.  
\EndFor
\ForAll{${\cal v}_{\cal i} \in {\cal V}$}
\State Create action node ${\cal v}_{\cal i}$ with pins named \textit{OperandIn} and \textit{OperandOut}.  Annotate it with the duration $d_{\cal i}$.
\State Create a central buffer $s_i$ and add it to $S_{\cal A}$.
\State Create a object flow arc $f_{ii} \in {\cal F}_o$ from action node ${\cal v}_{\cal i}$ to central buffer ${s}_{\cal i}$.  Annotate it with a weight $|Succ(v_i)|$.
\ForAll{$s_{\cal l} \in S_{\cal R}$}
\State Create an object flow arc $f_{{\cal l}i} \in {\cal F}_o$ (in black) from the operand requirement central buffer $s_{\cal l}$ to the \textit{OperandIn} pin of action node ${\cal v}_{\cal i}$ and annotate it with a weight ${\cal r}_{\cal il}$. 
\EndFor
\If{Operands are renewable}
\ForAll{$s_{\cal l} \in S_{\cal R}$}
\State Create an object flow arc $f_{i{\cal l}} \in {\cal F}_o$ (in red) from the \textit{OperandOut} pin of action node $v_i$ to the operand requirement central buffers $s_{\cal l}$ and annotate it with a weight ${\cal r}_{\cal il}$.  
\EndFor
\EndIf
\EndFor
\ForAll{${\cal a}_{\cal ij} \in {\cal A}$}
\State Create an object flow $f_{\cal ij} \in {\cal F}_o$ from central buffer $s_{\cal i}$ to action node ${\cal v}_{\cal j}$.  Annotate it with a weight equal to one.
\EndFor
\ForAll{${\cal v}_{\cal i} \in {\cal V} | ^{\bullet}{\cal v}_{\cal i} = \emptyset$}
\State Create a control flow arc $f \in {\cal F}_c$ from the ``Project Start" initial node to action node ${\cal v}_{\cal i}$.  Annotate it with a weight equal to one.
\EndFor
\ForAll{$s_{\cal l} \in S_{\cal A} | s_i^{\bullet} = \emptyset$}
\State Create a control flow arc $f \in {\cal F}_c$ from the central buffer $s_i$ to the ``Finish Project" action node.  Annotate it with a weight equal to one.
\EndFor
\State \textbf{return} ACT diagram ${\cal N}_{ACT} = \{\cal{V},S_{\cal R}\cup S_{\cal A},\mathcal{F}_o,\mathcal{F}_c,\text{Project Start},\text{Finish Project},\textsf{Terminal}\}$.
\EndProcedure
\end{algorithmic}
\end{algorithm}

Once a correspondence between the elements of the AoN project network and the BDD of the HFGT meta-architecture has been established, Algorithm \ref{Alg:AoN-to-ACT} is used to create a SysML activity diagram (ACT) from its associated AoN project network.  Each step of Alg. \ref{Alg:AoN-to-ACT} is elaborated for clarity.
\begin{enumerate}
\item In Line 1, the function is defined from the data in the AoN project network.  Alternatively, the equivalent data can be pulled from the engineering system BDD.  
\item In Line 2, the project start initial node serves a similar role to the dummy start activity $v_o$.
\item In Line 3, the finish project activity serves a similar role to the dummy end activity $v_f$.
\item In Line 4, the for loop is required to account for projects with many types of operand requirements.  
\item In Line 5, each newly created central buffer manages each operand requirement (i.e., originally defined in the RCPSP as a resource requirement).
\item In Line 6, the loop over operand requirements is closed.  
\item In Line 7, the for loop only includes nodes of the AoN project network and neglects the dummy start and end nodes.  
\item In Line 8, each activity node in the RCPSP AoN network maps 1-to-1 to an action node of the same name in the ACT.  For simplicity, the notation ${\cal v}_{\cal i}$ is overloaded to describe both the AoN project network activity nodes as well as the activity diagram action nodes.  
\item In Line 9, each activity node in the RCPSP AoN network maps to 1-to-1 to a central buffer indicating the completion of this activity. For example, the activity ``Perform Activity A" gains a central buffer called ``Activity A Complete".   
\item In Line 10, this action node and the central buffer must then be connected 1-to-1 for logical coherence.  Importantly, the arc weight is equal to the number of successor nodes to $v_i$. 
\item In Line 11, a for loop is initiated over the set of central buffers pertaining to operand requirements.  
\item In Line 12, the creation of an object flow arc from the operand requirements central buffer to each action node means that the non-renewable operand requirement is modeled directly into the ACT (and not just as an operational parameter that later appears in an RCPSP constraint).  
\item In Line 13, the for loop over the operand requirement central buffers is closed.  
\item In Line 14, an if condition is initiated for the case of renewable (rather than non-renewable) operands.  
\item In Line 15, a for loop is initiated over the set of central buffers pertaining to operand requirements.  
\item In Line 16, similarly to Line 12, and if needed, the renewable operand requirement is modeled directly into the ACT.  
\item In Line 17, the for loop over the operand requirement central buffers is closed.  
\item In Line 18, the if condition is closed. 
\item In Line 19, the loop over activity nodes is closed.  
\item In Line 20, the for loop includes all of the arcs of the AoN project network. 
\item In Line 21, the precedence between a completed action and its succeeding action is established.  
\item In Line 22, the loop over AoN project network arcs is closed. 
\item In Line 23, the for loop includes all of the action nodes that have no preceding central buffer.  The notation $^{\bullet}()$ means ``preset of" as understood in the Petri net literature\cite{Murata:1989:01}.  
\item In Line 24, each action node with no preset receives an arc from the ``Project Start" initial node.  
\item In Line 25, the for loop over these action nodes is closed. 
\item In Line 26, the for loop includes all of the central buffers with no succeeding action.  The notation $()^{\bullet}$ means ``postset of" as understood in the Petri net literature\cite{Murata:1989:01}.  
\item In Line 27, each central buffer with no postset receives an arc to the ``Final Project" action.  
\item In Line 28, the for loop over these central buffers is closed. 
\item In Line 29, the activity diagram is returned as a tuple that includes the activity nodes, the central buffers, object flow arcs, the control flow arcs, the project start initial node, the finish project action, and the terminal node.
\item In Line 30, Alg. \ref{Alg:AoN-to-ACT} is concluded.  
\end{enumerate}

In the constructed ACT, the \emph{central buffer} represents the pool of operand tokens for a type of operand.  At the start of each activity, required tokens flow from the central buffer into the corresponding action node. These inflows are represented by black object flows, which are common to both renewable and non-renewable operands. The distinction becomes apparent upon completion: for renewable operands, red return flows transfer the seized tokens back to the central buffer, thereby restoring availability for future tasks. In contrast, non-renewable operands are consumed irreversibly, with no return flow. Control flows enforce the precedence relations among activities, while the object flows encode operand usage and release. Together, this structure allows the ACT to capture both temporal sequencing and the contrasting dynamics of renewable versus non-renewable operand constraints.

\subsection{Hetero-functional Graph Theory Definitions}\label{sec:HFGT_SecIII}
Once an arbitrary AoN project has been described from an MBSE perspective, the correspondence between the RCPSP and HFGT can be established.  More specifically, each of the definitions in Sec. \ref{Sec:HFGTIntro} is discussed in this context.

For Def.~\ref{Defn:D1 System Operand}, the \emph{operand} is the (one) project itself; viewed as it evolves from one activity to another over discrete time\cite{Hosseini:2025:ISC-J55}.  Additionally, the ``operand requirements" of the RCPSP (originally called resource requirements) imply an additional set of operands.  As noted in the $\alpha$-branch of the machine-scheduling taxonomy (Sec.~\ref{Sec:RCPSP-Background} and Sec.~\ref{MBSE_SecIII}), these operands are categorized as renewable, non-renewable, cumulative, etc.  These additional operands are, however, secondary in that they are part of the overall project (as an operand).  Note that in Fig. \ref{Fig:LFESMetaArchitecture}, the aggregation arrow from the operand block back to itself allows one or more operands to exist within another operand.  The elaboration of Defn. \ref{Defn:OperandNet} below discusses how these secondary operands are incorporated into the project's operand net.

For Defs.~\ref{Defn:D2 System Process}-\ref{Defn:ESN-STF}, they are not required by the RCPSP and are, therefore, not instantiated by the AoN project network in Fig. \ref{fig:SampleProjectNetwork}.  This perhaps surprising statement reflects the subtlety in HFGT modeling and is best understood by example.  Consider a manufacturing system with two identical operands: ``Metal Block A", and ``Metal Block B".  Each metal block has two states ``Without Hole Feature" and ``With Hole Feature" and one state transition that describes the flip from the former state to the latter state.  Additionally, this manufacturing system has two system resources: ``Drill Press A" and ``Drill Press B".  Both of these system resources can carry out exactly one system process called ``drill metal block".  This system process must not be confounded with a state transition that flips a metal block from its ``Without Hole Feature" state to its ``With Hole Feature" state.  The execution of the system process ``drill metal block" (at either drill press) does not necessitate a state transition in a specific metal block because there are two metal blocks that can receive this system process.  Similarly, the state transition in a specific metal block does not specify whether a system process was carried out by Drill Press A or B.  Because the relationship between state transitions in an operand is not related to system processes in a one-to-one fashion, they must be modeled as distinct elements.  Returning to correspondence betwen the RCPSP and HFGT, Defs.~\ref{Defn:D2 System Process}-\ref{Defn:ESN-STF} are elaborated:
\begin{itemize}
\item For Defn. \ref{Defn:D2 System Process}, there are no system processes that evolve or execute the project forward (as the system's only operand).  The project's activities that describe the transition from one project state to another must not confounded with the system processes.  
\item For Defn. \ref{Defn:SystemResource}, the are no system resources because the RCPSP forces $\alpha_1=\emptyset$.  
\item For Defn. \ref{Defn:D4 Buffer}, there are no buffers because there are no resources.  
\item For Defn. \ref{Defn:Capability}, there are no capabilities because there are no system processes or resources.  
\item Consequently, for Defn. \ref{Defn:D6 HFIT -ve}--\ref{Defn:ESN-STF}, the hetero-functional incidence tensors, engineering system net, and engineering system net state transition function collapse to triviality.  
\end{itemize}

\begin{figure}
    \centering
    \includegraphics[width=\linewidth]{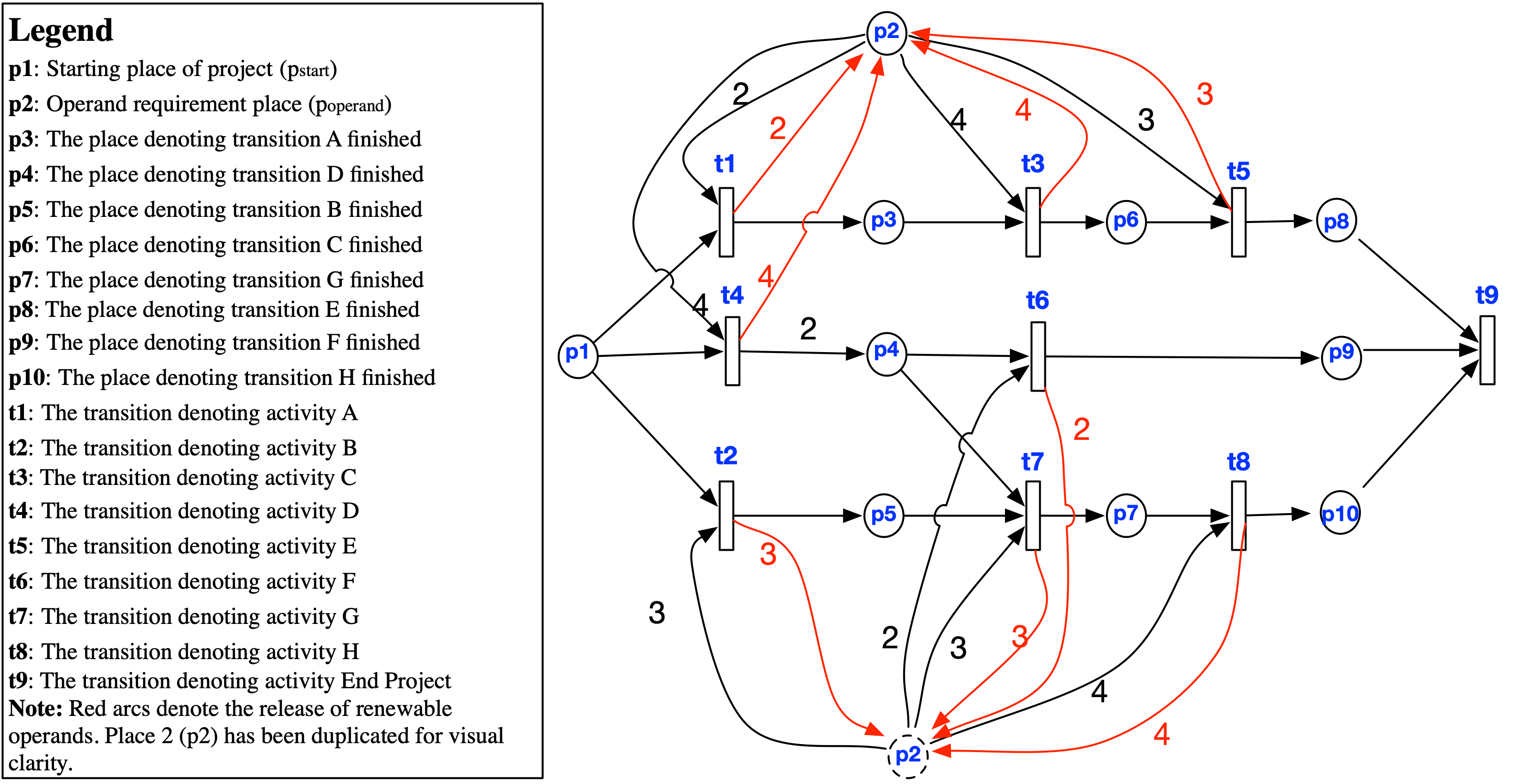}
    \caption{Project Operand Net}
    \label{fig:OperandNet}
\end{figure}

\begin{algorithm}
\caption{ACT-to-Operand Net Construction for Renewable/Non-Renewable Operands}\label{Alg:ACT-to-OperandNet}
\begin{algorithmic}[1]
\Procedure{ACTtoOperandNet}{${\cal N}_{ACT} = \{\cal{V},S_{\cal R}\cup S_{\cal A},\mathcal{F}_o,\mathcal{F}_c,\text{Project Start},\text{Finish Project},\textsf{Terminal}\}$}
\State Transform ``Start Project'' into an initial place $s_{1}$ and add it to $S$. Give it an initial marking $Q_{1l}=|s_1^{\bullet}|$
\State Delete the \textsf{Terminal} node and all arcs incident to it.
\ForAll{${\cal v}_{\cal i} \in \cal{V}$}
  \State Transform action ${\cal v}_{\cal i}$ into transition $\epsilon_{\cal i}$ and add it to ${\cal E}_l$.
\EndFor
\State Transform the action ``Finish Project'' into final transition $\epsilon_{\textsf{fin}}$.
\ForAll{$s_{\cal l} \in \cal{S}_{\cal R}$}
  \State Transform central buffer $s_{\cal l}$ into place $s_{\cal l}$ and add it to $S_{\cal R}$. 
  \If{If Operands are Renewable}
  \State Give place $s_{\cal l}$ an initial marking $Q_{{\cal l}l}=C_{\cal lRr}$.
  \Else 
  \State Give place $s_{\cal l}$ an initial marking $Q_{{\cal l}l}=C_{\cal lRn}$.
  \EndIf
\EndFor
\ForAll{$s_{\cal l} \in \cal{S}_{\cal A}$}
  \State Transform central buffer $s_{\cal l}$ into place $s_{\cal l}$ and add to $S_{\cal A}$. Give it an initial marking $Q_{{\cal l}l}=0$.
\EndFor
\Statex\textbf{Flow-to-arc rule:}
\ForAll{$f_{xy}\in \mathcal{F}_o\cup\mathcal{F}_c, y\neq\textsf{Terminal}$}
\State Replace the control or object flow $f_{xy}$ with an arc $\mathbf{M}_{xy} \in \textbf{M}_{l}$ (from a place to a transition or a transition to a place) that \emph{preserves direction, weight, and color (i.e., black inflows, red renewable returns)}. 
\EndFor
\State \textbf{return} Operand Net ${\cal N}_{l}= \{s_{1}\cup \cal{S}_{\cal R} \cup \cal{S}_{\cal A}, {\cal E}_{l}, \textbf{M}_{l}, W_{l}, Q_{l}\}$
\EndProcedure
\end{algorithmic}
\end{algorithm}

For Defn. \ref{Defn:OperandNet}, it becomes clear that the operand net becomes the central definition for creating the correspondence between the RCPSP and HFGT.  Once the AoN project network (e.g., Fig. \ref{fig:SampleProjectNetwork}) is transformed into an activity diagram via Alg. \ref{Alg:AoN-to-ACT}, it is then straightforwardly transformed into an operand net with Alg. \ref{Alg:ACT-to-OperandNet}.  Note that the secondary operands, discussed under Defn. \ref{Defn:D1 System Operand}, become the places $S_{\cal R}$ in this new operand net.  As a result, the activity diagram shown in Fig. \ref{fig:SampleProjectACT} becomes the operand net shown in Fig. \ref{fig:OperandNet}.  Finally, all of the arcs of the activity diagram \emph{with their associated weights} are retained as arcs in the operand net without change.  Note that in order to reflect the role of renewable operands, the red arcs to the required operand buffer in Fig. \ref{fig:SampleProjectACT} are depicted similarly in Fig. \ref{fig:OperandNet}.  The primary advantage of the operand net over the AoN project network is the explicit representation of project state over time. As discussed in Sec. \ref{Sec:Discussion}, this state-based representation can substantially aid project monitoring and control. 

For Defn. \ref{Defn:OperandNet-STF}, the operand net state transition function serves to track and evolve the project state at each time step from its initial place to its final transition as its \emph{makespan}.  For example, the general form of the operand net state transition function in Eq. \ref{EQ:ON-STF} becomes Eq. \ref{Eq:HFNMCF_OperandNet-STF1_Matrix} when applied to the operand net shown in Fig. \ref{fig:OperandNet}.  (To describe renewable operands, the red entries in the operand net incidence matrix correspond to the red arcs in Fig. \ref{fig:OperandNet}.)
\begin{align}\label{Eq:HFNMCF_OperandNet-STF1_Matrix}
Q_{Sl}[k+1]&=Q_{Sl}[k]
+\underbrace{\begin{bmatrix}
0 & 0 & 0 & 0 & 0 & 0 & 0 & 0 & 0\\
\textcolor{red}{2} & \textcolor{red}{3} & \textcolor{red}{4} & \textcolor{red}{4} & \textcolor{red}{3} & \textcolor{red}{2} & \textcolor{red}{3} & \textcolor{red}{4} & \textcolor{red}{0}\\
1 & 0 & 0 & 0 & 0 & 0 & 0 & 0 & 0\\
0 & 0 & 0 & 2 & 0 & 0 & 0 & 0 & 0\\
0 & 1 & 0 & 0 & 0 & 0 & 0 & 0 & 0\\
0 & 0 & 1 & 0 & 0 & 0 & 0 & 0 & 0\\
0 & 0 & 0 & 0 & 0 & 0 & 1 & 0 & 0\\
0 & 0 & 0 & 0 & 1 & 0 & 0 & 0 & 0\\
0 & 0 & 0 & 0 & 0 & 1 & 0 & 0 & 0\\
0 & 0 & 0 & 0 & 0 & 0 & 0 & 1 & 0
\end{bmatrix}}_{\text{output incidence } M^{+}}
U_{l}^+[k]\Delta T
-\underbrace{\begin{bmatrix}
1 & 1 & 0 & 1 & 0 & 0 & 0 & 0 & 0\\
2 & 3 & 4 & 4 & 3 & 2 & 3 & 4 & 0\\
0 & 0 & 1 & 0 & 0 & 0 & 0 & 0 & 0\\
0 & 0 & 0 & 0 & 0 & 1 & 1 & 0 & 0\\
0 & 0 & 0 & 0 & 0 & 0 & 1 & 0 & 0\\
0 & 0 & 0 & 0 & 1 & 0 & 0 & 0 & 0\\
0 & 0 & 0 & 0 & 0 & 0 & 0 & 1 & 0\\
0 & 0 & 0 & 0 & 0 & 0 & 0 & 0 & 1\\
0 & 0 & 0 & 0 & 0 & 0 & 0 & 0 & 1\\
0 & 0 & 0 & 0 & 0 & 0 & 0 & 0 & 1
\end{bmatrix}}_{\text{input incidence } M^{-}}
U_{l}^-[k]\Delta T
\end{align}
\noindent Together, Eqs. \ref{Eq:HFNMCF_OperandNet-STF1_Matrix} and Eq. \ref{EQ:ON-STF} reveal several important points.
\begin{itemize}
\item The input firing vector $U_l^-[k]$ tracks when each project activity begins. 
\item The output firing vector $U_l^+[k]$ tracks when each project activity ends.
\item The marking vector $Q_{Sl}[k]$ tracks when the preconditions of each project activity have been sufficiently met.
\item At the project start:
\begin{itemize}
\item In the case of renewable operands, the marking vector $Q_{Sl}[1]=[|s_i^{\bullet}|; C_{l{\cal R}r}; \mathbf{0}]$ so that there are sufficient tokens in $s_1$, $S_{\cal R}$, and $S_{\cal A}$ respectively for the project to begin.  
\item In the case of non-renewable operands, the marking vector $Q_{Sl}[1]=[|s_i^{\bullet}|; C_{l{\cal R}n}; \mathbf{0}]$ so that there are sufficient tokens in $s_1$, $S_{\cal R}$, and $S_{\cal A}$ respectively for the project to begin.  
\end{itemize}
\item At the project end, 
\begin{itemize}
\item In the case of renewable operands, the marking vector $Q_{Sl}[K+1]=\mathbf{0}$, emphasizing that all tasks have been completed and that all non-renewable operands have been consumed.  
\item In the case of non-renewable operands, the marking vector $Q_{Sl}[K+1]=[0; C_{\cal Rr}; \mathbf{0}]$, emphasizing that all tasks have been completed but that all renewable operands have been restored to their original values.  
\end{itemize}
\item The marking vector $Q_{{\cal E}l}[k]$ tracks when each project activity is still ongoing.  
\end{itemize}
Together, all four of these quantities constitute a \emph{mutually exclusive and collectively exhaustive} description of a project's logical state.

\subsection{RCPSP Specialization of the HFNMCF Problem to the RCPSP}\label{Sec:HFNMCF_RCPSP}

Given the correspondence between the RCPSP and the HFGT definitions in the previous section, the HFNMCF problem can be solved straightforwardly.  In this respect, it becomes clear that the complete form of the HFNMCF problem found in Eq. \ref{Eq:HFNMCF_ObjFunc}--\ref{Eq:HFNMCF_DeviceModels2} is overly general in the context of the RCPSP.  Instead, the following context-dependent limitations are inherited from the RCPSP.  
\begin{itemize}
\item For simplicity, and without loss of generality, define $\Delta T\;\equiv\;1$
\item $L = \{l\}$ because only one project is being scheduled.  
\item  $x[k] = \begin{bmatrix}Q_{Sl}; U^-_l ; U^+_l \end{bmatrix}[k] \forall k \in \{1, \dots, K\}$ because the RCPSP does not have an engineering system net or its associated state transition function.  
\item $U^-_l[k] \in \{0,1\}^{|{\cal E}_l|}\forall k \in \{1, \dots, K\}$ because initiation of a project transition is a binary condition.  
\item $U^+_l[k] \in \{0,1\}^{|{\cal E}_l|}\forall k \in \{1, \dots, K\}$ because the completion of a project transition is a binary condition.  
\item $Q_{Sl}[k] \in \mathbb{Z}^{+|S_l|}\forall k \in \{1, \dots, K\}$ because the operand net incidence matrix consists of exclusively positive values.  
\item In Eq. \ref{Eq:HFNMCF_ObjFunc}, $F_{QP}[k]=0 \forall k$ because the RCPSP utilizes a linear rather than quadratic objective function.  For a minimal project makespan, $f_{QP}[k]=[\mathbf{0}; k\cdot e_n; \mathbf{0}] \quad \forall k$ where the first column zero corresponds to the $Q_{sl}[k]$, the $n^{th}$ elementary basis vector $e_n$ corresponds to $U^{-}_{l}[k]$, and the last column zero corresponds to $U^{+}_{l}[k]$.  
\item Eqs. \ref{Eq:HFNMCF_ESN-STF1} --\ref{Eq:HFNMCF_DurationConstraint} collapse to triviality because the RCPSP does not have an engineering system net or its associated state transition function.  
\item Eq. \ref{Eq:HFNMCF_OperandNet-STF1} (or more specifically Eq. \ref{Eq:HFNMCF_OperandNet-STF1_Matrix}) is retained without change.
\item Eq. \ref{Eq:HFNMCF_OperandNet-STF2} can be omitted from the optimization because no additional constraints are placed on the $Q_{{\cal E}l_i}[k]$ variables. Therefore, they can be easily calculated after the optimization problem has been solved (as demonstrated in Sec. \ref{Sec:Solution_HFNMCF_RCPSP}).
\item Eq. \ref{Eq:HFNMCF_OperandNetDurationConstraint} is retained without change.  
\item Eqs. \ref{Eq:HFNMCF_SyncPlus}-\ref{Eq:HFNMCF_ESN-Exogenous} also collapse to triviality because the RCPSP does not have an engineering system net or its associated state transition function.  
\item Eq. \ref{Eq:HFNMCF_OperandNet-Exogenous} is not required because the RCPSP does not place any exogenous constraints on which activities are executed or when.  $E_{L_p}=E_{L_n}=0$.  
\item Eq. \ref{Eq:HFNMCF_InitCond} is retained, but without the $Q_B[1]$ and $Q_{\cal E}[1]$ variables. As explained in the previous section, $Q_{Sl}[1]=[|s_i^{\bullet}|; C_{l{\cal R}r}; \mathbf{0}]$ in the case of renewable operands.  $Q_{Sl}[1]=[|s_i^{\bullet}|; C_{l{\cal R}r}; \mathbf{0}]$ in the case of non-renewable operands.  
\item Eq. \ref{Eq:HFNMCF_FinalCond} is retained, but without the $Q_B[K+1]$, $Q_{\cal E}[K+1]$, and $U^-[K+1]$ variables.  As explained in the previous section, $Q_{Sl}[K+1]=[0; C_{\cal Rr}; \mathbf{0}]$ in the case of renewable operands.  $Q_{Sl}[K+1]=\mathbf{0}$ in the case of non-renewable operands.
\item In Eq. \ref{Eq:HFNMCF_Capacity}, $\underline{E}_{CP}=0$ to impose non-negativity of decision variables.  $D(X)=DX$ where $D=I$ so that lower bounds are placed on the decision variables individually.  Finally, as Sec. \ref{Sec:Discussion} discusses, the upper bound constraint is not required.  
\item Eqs. \ref{Eq:HFNMCF_DeviceModels1} and \ref{Eq:HFNMCF_DeviceModels2} are not required because the RCPSP does not have device model constraints.  
\end{itemize}
 
As a result, the HFNMCF problem when applied to the RCPSP context becomes: 
\begin{subequations}
\begin{alignat}{2}
Z_{RCPSP} =& \sum_{k=1}^{K-1} [k\cdot e_n]^T U^{-}_{l}&&[k] \label{Eq:ObjFunc_HFNMCF_RCPSP} \\
\text{s.t.} -Q_{Sl}[k+1]+Q_{Sl}[k]+{M}_{l}^+U_{l}^+[k] - {M}_{l}^-U_{l}^-[k] = &0 && \forall k \in \{1, \dots, K\} \label{Eq:OperandNet-STF1_HFNMCF_RCPSP}\\
-U_{xl}^+[k+d_{xl}]+ U_{xl}^-[k] = &0 &&  \forall k\in \{1, \dots, K\}, \: \forall x\in \{1, \dots, |{\cal E}_{l}|\} \label{Eq:OperandNetDurationConstraint_HFNMCF_RCPSP}\\ 
Q_{Sl}[1] =& C_{{Sl}1} \label{Eq:InitCond_HFNMCF_RCPSP}\\ 
\begin{bmatrix}  Q_{Sl} ; U_l^- \end{bmatrix}[K+1] =   &\begin{bmatrix} C_{{Sl}K} ; \mathbf{0} \end{bmatrix} \label{Eq:FinalCond_HFNMCF_RCPSP}\\ 
0 \leq \begin{bmatrix}Q_{Sl}; U^-_l ; U^+_l \end{bmatrix}[k] \;\;\;\; &  && \forall k \in \{1, \dots, K\} \label{Eq:Capacity_HFNMCF_RCPSP}
\end{alignat}
\end{subequations}
where:
\begin{itemize}
\item Eq.~\ref{Eq:ObjFunc_HFNMCF_RCPSP} minimizes the project makespan.  
\item Eq.~\ref{Eq:OperandNet-STF1_HFNMCF_RCPSP} applies the operand net state transition function as defined in Eq. \ref{Eq:HFNMCF_OperandNet-STF1} (and Eq. \ref{Eq:HFNMCF_OperandNet-STF1_Matrix} more specifically).  
\item Eq.~\ref{Eq:OperandNetDurationConstraint_HFNMCF_RCPSP} enforces transition durations so that an activity ends $d_{\cal i}$ timesteps after it starts.  
\item Eq.~\ref{Eq:InitCond_HFNMCF_RCPSP} establishes the initial conditions of the project operand net, ensuring that the initial node has enough tokens to enable its postset transitions and that the operand requirement places $S_{\cal R}$ have enough operands at the project start.  
\item Eq.~\ref{Eq:FinalCond_HFNMCF_RCPSP} establishes the final conditions of the project operand net, ensuring that the project ends by exhausting all of its non-renewable operands or fully replenishing its renewable ones. 
\item Finally, Eq.~\ref{Eq:Capacity_HFNMCF_RCPSP} imposes capacity constraints on the decision variable at each timestep.  The upper bound on the operand net places represents the availability of renewable and non-renewable operands.  The upper bound on the operand transitions ensures that each is only executed once at a given time.    
\end{itemize}
Together, Eqs. \ref{Eq:ObjFunc_HFNMCF_RCPSP}-\ref{Eq:Capacity_HFNMCF_RCPSP} describe an expressive RCPSP formulation within the HFNMCF problem.

\subsection{Solution of RCPSP Specialization of the HFNMCF Problem}\label{Sec:Solution_HFNMCF_RCPSP}
Given the explanation of the previous section, for the case of renewable operands, the RCPSP specialization of HFNMCF problem found in Eq. \ref{Eq:ObjFunc_HFNMCF_RCPSP}--\ref{Eq:Capacity_HFNMCF_RCPSP} is solved to optimality.  The optimal values of the column vectors $U_l^{-*}[k]$ and $U_l^{+*}[k]$ are horizontally concatenated and presented below.  
\begin{footnotesize}
\begin{align}\label{Eq:OptimalU}
U_{l}^{-*} =
\left[
\begin{array}{*{17}{r|}r}
1 &   &   &   &   &   &    &    &    &    &    &    &    &    &    &    &   &  \\\hline
  &   & 1 &   &   &   &    &    &    &    &    &    &    &    &    &    &   &  \\\hline
  &   &   &   & 1 &   &    &    &    &    &    &    &    &    &    &    &   &  \\\hline
1 &   &   &   &   &   &    &    &    &    &    &    &    &    &    &    &   &  \\\hline
  &   &   &   &   &   &    & 1  &    &    &    &    &    &    &    &    &   &  \\\hline
  &   &   &   &   &   &    & 1  &    &    &    &    &    &    &    &    &   &  \\\hline
  &   &   &   &   &   &    &    &    & 1  &    &    &    &    &    &    &   &  \\\hline
  &   &   &   &   &   &    &    &    &    &    &    &    & 1  &    &    &   &  \\\hline
  &   &   &   &   &   &    &    &    &    &    &    &    &    &    & 1  &   &  \\
\end{array}
\right],
U_{l}^{+*} =
\left[
\begin{array}{*{17}{r|}r}
  &   & 1 &   &   &   &    &    &    &    &    &    &    &    &    &    &   &  \\\hline
  &   &   &   &   &   &    &    &    & 1  &    &    &    &    &    &    &   &  \\\hline
  &   &   &   &   &   &    &  1 &    &    &    &    &    &    &    &    &   &  \\\hline
  &   &   &   & 1 &   &    &    &    &    &    &    &    &    &    &    &   &  \\\hline
  &   &   &   &   &   &    &    &    &    &    &    &    &    &    &  1 &   &  \\\hline
  &   &   &   &   &   &    &    &    &    &    &    &    &    &  1 &    &   &  \\\hline
  &   &   &   &   &   &    &    &    &    &    &    &    &    &  1 &    &   &  \\\hline
  &   &   &   &   &   &    &    &    &    &    &    &    &    &    &  1 &   &  \\\hline
  &   &   &   &   &   &    &    &    &    &    &    &    &    &    &  1 &   &  \\
\end{array}
\right]
\end{align}
\end{footnotesize}
\noindent Consequently, each row corresponds to a transition, and each column corresponds to a time index.  Eq. \ref{Eq:OptimalU} shows when each project activity (or transition) begins with $U_l^{-*}[k]$ and ends with $U_l^{+*}[k]$.  Quite notably, the optimal solution $U_{l}^{-*}$ (from the HFNMCF problem) is equivalent to ${\cal x}^{*}$ (in the RCPSP).  In the meantime, the HFNMCF problem explicitly produces the project finish times in $U_l^{+*}[k]$ whereas the RCPSP states them implicitly in Eq. \ref{Eq:RCPSP_constraint2}.  Additionally, the optimal value $Q^{*}_{{\cal E}l[k]}$ can be calculated after the optimization via Eq. \ref{Eq:HFNMCF_OperandNet-STF2}.  
\begin{footnotesize}
\begin{align}\label{Eq:OptimalQEl}
Q^{*}_{{\cal E}l} =
\left[
\begin{array}{*{18}{r|}r}
1 & 1 &   &   &   &   &    &    &    &    &    &    &    &    &    &    &   &  &   \\\hline
  &   & 1 & 1 & 1 & 1 & 1  & 1  & 1  &    &    &    &    &    &    &    &   &  &   \\\hline
  &   &   &   & 1 & 1 & 1  &    &    &    &    &    &    &    &    &    &   &  &   \\\hline
1 & 1 & 1 & 1 &   &   &    &    &    &    &    &    &    &    &    &    &   &  &   \\\hline
  &   &   &   &   &   &    & 1  & 1  & 1  & 1  & 1  & 1  & 1  & 1  &    &   &  &   \\\hline
  &   &   &   &   &   &    & 1  & 1  & 1  & 1  & 1  & 1  &    &    &    &   &  &   \\\hline
  &   &   &   &   &   &    &    &    & 1  & 1  & 1  & 1  &    &    &    &   &  &   \\\hline
  &   &   &   &   &   &    &    &    &    &    &    &    & 1  & 1  &    &   &  &   \\\hline
  &   &   &   &   &   &    &    &    &    &    &    &    &    &    &    &   &  &   \\
\end{array}
\right]
\end{align}
\end{footnotesize}
\noindent Again, recall that Sec. \ref{Sec:HFNMCF_RCPSP} explains that Eq. \ref{Eq:HFNMCF_OperandNet-STF2} does not need to be retained in the RCPSP-specialization of the HFNMCF problem because it can be calculated after the solution of the optimization program.  Quite notably, the optimal value $Q^{*}_{{\cal E}l[k]}$ is equivalent to the results found in Table \ref{tab:RCPSP-schedule-renewable} when presented in binary form.  This shows that the HFNCMF problem explicitly describes when each task is ongoing while this state must be inferred by the RCPSP from the quantity $\sum_{\kappa=k-d_i}^{\kappa=k} {\cal x}_{{\cal i}\kappa} \forall k \in \{1, \dots, K\}$.  Finally, the optimal values of $Q^{*}_{Sl}[k]$ are horizontally concatenated and presented below.  
\begin{footnotesize}
\begin{align}\label{Eq:OptimalQSl}
Q_{S_l} =
\left[
\begin{array}{*{18}{r|}r}
3 & 1 & 1 &   &   &   &   &   &   &   &   &   &   &   &   &   &   &   &  \\\hline
8 & 2 & 2 & 1 & 1 & 1 & 1 &  1&   &   &   &   &   &   & 1 & 1 & 8 & 8 & 8 \\\hline
  &   &   & 1 & 1 &   &   &   &   &   &   &   &   &   &   &   &   &   &   \\\hline
  &   &   &   &   & 2 & 2 & 2 & 1 & 1 &   &   &   &   &   &   &   &   &   \\\hline
  &   &   &   &   &   &   &   &   &   &   &   &   &   &   &   &   &   &   \\\hline
  &   &   &   &   &   &   &   &   &   &   &   &   &   &   &   &   &   &   \\\hline
  &   &   &   &   &   &   &   &   &   &   &   &   &   &   &   &   &   &   \\\hline
  &   &   &   &   &   &   &   &   &   &   &   &   &   &   &   &   &   &   \\\hline
  &   &   &   &   &   &   &   &   &   &   &   &   &   & 1 & 1 &   &   &   \\\hline
  &   &   &   &   &   &   &   &   &   &   &   &   &   &   &   &   &   &   
\end{array}
\right]
\end{align}
\end{footnotesize}
\noindent Tokens in the first place $p_{start}$ show the number of initial activities that have yet to be started.  Here, one token remains in $p_{start}$ at $k=2$ and $k=3$ because the second transition does not begin until $k=3$.  Meanwhile, tokens in the second row, $p_{operand}$, denote the availability of (renewable) operands at every time step. Here, the eight renewable operands are engaged through the parallel execution of simultaneous transitions, but are ultimately returned to $p_{operand}$ as each of these transitions finishes.  Lastly, the presence of nonzero values in the remaining places indicates tokens that are \emph{waiting} to be utilized by succeeding transitions.  This indicates inefficiency in the project schedule because there is slack time between the completion time of one project activity and the start time of its successor activity.  In this case, $p_3$, $p_4$, and $p_9$ are involved in these waiting times.  Again, the RCPSP formulation does not explicitly state these waiting times, although they are present \emph{implicitly}.  When Eq. \ref{Eq:RCPSP_constraint2} is rewritten as an equality constraint with slack times $T_{\cal ij}$, 
\begin{align}\label{Eq:RCPSP_constraint2_w/Slack}
\left(\sum_{k = 1}^{K} k \cdot {\cal x}_{ik} + d_{\cal i} \right) + T_{\cal ij} = \sum_{k = 1}^{K} k \cdot {\cal x}_{jk} \quad \forall ({\cal i}, {\cal j}) \in {\cal A} 
\end{align}
Because each place $s_{\cal l}$ exists along an AoN project network arc ${\cal a}_{ij} \in {\cal A}$, it becomes clear that the nonzero durations in $Q_{S_l}$ map one-to-one to these slack times.  

In all, the results presented in Eqs. \ref{Eq:OptimalU}-\ref{Eq:OptimalQSl} provide numerical evidence that the RCPSP specialization of the HFNMCF problem (in Eq. \ref{Eq:ObjFunc_HFNMCF_RCPSP}-\ref{Eq:Capacity_HFNMCF_RCPSP}) is equivalent to the RCPSP stated in Eqs. \ref{Eq:RCPSP_Objective}-\ref{Eq:RCPSP_constraint3} for the cases of renewable operands.  The results for the case of the non-renewable operands are presented in Eq. \ref{Eq:OptimalU2}-\ref{Eq:OptimalQ2}.  Again, they show equivalence to their RCPSP counterparts in Eq. \ref{Eq:RCPSP_OptimalXNonRenewable} and Table \ref{tab:RCPSP-schedule-nonrenewable}.  The results from both examples are elaborated upon in the discussion section that follows.  

\section{Discussion: The Relative Merits of the RCPSP Specialization of the HFNMCF Problem}\label{Sec:Discussion}

The previous section concretely demonstrated the relationship between the RCPSP and the RCPSP specialization of the HFNMCF problem.  In both cases of renewable and non-renewable operands, the solution to the RCPSP was equivalent to the RCPSP specialization of the HFNMCF problem.  In effect, both minimize a project objective (e.g., project makespan) subject to precedence constraints, operand requirement constraints, and project activity completion constraints.  Consequently, while the numerical evidence in the previous section is compelling from a pedagogical perspective, the pattern of results points to a more general result.  
\begin{thm}\label{Thm:HFNMCF-RCPSP}
The RCPSP specialization of the HFNMCF problem in Eqs. \ref{Eq:ObjFunc_HFNMCF_RCPSP}--\ref{Eq:Capacity_HFNMCF_RCPSP} is a generalization of the RCPSP in Eqs. \ref{Eq:RCPSP_Objective}--\ref{Eq:RCPSP_constraint4}.
\end{thm}
\begin{proof}
The strategy of the proof is to first recognize the commonality of decision variables and indices, and then show how Eqs. \ref{Eq:ObjFunc_HFNMCF_RCPSP}--\ref{Eq:Capacity_HFNMCF_RCPSP} necessitate Eqs. \ref{Eq:RCPSP_Objective}--\ref{Eq:RCPSP_constraint4}.
\begin{itemize}[label=$\bullet$]
\item By Alg. \ref{Alg:AoN-to-ACT} and \ref{Alg:ACT-to-OperandNet}, the node index $\cal i$ is equivalent to the transition index $x$.  ${\cal i}=x$.
\item By definition, ${\cal x}_{{\cal i}k} \;\equiv\; U^{-}_{xl}[k] \quad \forall x,k$.  (For clarity, $\cal x$ is the decision variable in the RCPSP and $\chi$ is the operand net transition index as shown in the Nomenclature section.)
\item For Eq. \ref{Eq:RCPSP_Objective}, Eq. \ref{Eq:ObjFunc_HFNMCF_RCPSP} is an algebraically equivalent restatement in matrix form.  
\item For Eq. \ref{Eq:RCPSP_constraint1}, the state transition function in Eq. \ref{Eq:OperandNet-STF1_HFNMCF_RCPSP}, the initial condition in \ref{Eq:InitCond_HFNMCF_RCPSP}, and the final condition \ref{Eq:FinalCond_HFNMCF_RCPSP} guarantee a solution where $\sum_{k=1}^{K+1}U_{l}^-[k]=\mathds{1}$.  For the initial place $s_1$ (which has no preset transitions), when its postset transitions $s_1^{\bullet}$ are fired once, all the tokens are removed to meet the final condition.  For all of the places $s \in S_{\cal A}$, when its preset transitions and postset transitions are fired once, the marking vector remains unchanged.  The execution of Alg. \ref{Alg:AoN-to-ACT} and \ref{Alg:ACT-to-OperandNet} guarantees that each row of $M_l=M^+-M^-$ corresponding to the places $s \in S_{\cal A}$ sums to zero.  Similarly, for all of the places $s \in S_{\cal R}$, in the case of renewable operands, again the marking vector remains unchanged.  In contrast, in the case of non-renewable operands, the marking vector is driven down to zero as a final condition.  
\item For Eq. \ref{Eq:RCPSP_constraint2}, it is guaranteed by the state transition function in Eq. \ref{Eq:OperandNet-STF1_HFNMCF_RCPSP}, the duration constraint in Eq. \ref{Eq:OperandNetDurationConstraint_HFNMCF_RCPSP}, the initial condition in Eq. \ref{Eq:InitCond_HFNMCF_RCPSP} and the non-negativity constraint in Eq. \ref{Eq:Capacity_HFNMCF_RCPSP}.  First, algebraically rearrange Eq. \ref{Eq:OperandNet-STF1_HFNMCF_RCPSP} and rewrite it in scalar form using place index $\cal l$ and transition indices ${\cal i}$ and ${\cal j}$.  Substituting in Eq. \ref{Eq:OperandNetDurationConstraint_HFNMCF_RCPSP} yields:  
\begin{align}
\sum_{\cal i}{M}_{{\cal li}l}^+U_{{\cal i}l}^-[k+d_{\cal i}] - \sum_{\cal j}{M}_{{\cal lj}l}^-U_{{\cal j}l}^-[k] =& Q_{{\cal l}l}[k+1]-Q_{{\cal l}l}[k] \quad \forall k \in \{1, \dots, K\},\forall l \in \{1,\ldots,|S_{\cal A}|\}
\end{align}
Next, note that Eq. \ref{Eq:InitCond_HFNMCF_RCPSP} imposes $Q_{sl}[1]=0 \forall s_{\cal l} \in S_{\cal A}$, and Eq. \ref{Eq:Capacity_HFNMCF_RCPSP} imposes non-negativity on these same places.  Therefore, $U_{{\cal j}l}^-[k]$ must occur at the same time or after $U_{{\cal i}l}^+[k+d_{\cal i}]$ to prevent non-negativity on $Q_{{\cal l}l}[k]$.  
\item For Eq. \ref{Eq:RCPSP_constraint3}, it is guaranteed by the state transition function in Eq. \ref{Eq:OperandNet-STF1_HFNMCF_RCPSP}, the initial condition in Eq. \ref{Eq:InitCond_HFNMCF_RCPSP}, and the non-negativity lower bound in Eq. \ref{Eq:Capacity_HFNMCF_RCPSP}.  Combining Eqs. \ref{Eq:OperandNet-STF1_HFNMCF_RCPSP}, \ref{Eq:InitCond_HFNMCF_RCPSP} and Eq. \ref{Eq:Capacity_HFNMCF_RCPSP} and then rewriting in scalar form yields:  
\begin{align}
\sum_{\cal j}{M}_{{\cal lj}l}^-U_{{\cal j}l}^-[k] \leq & C_{\cal lRr} \quad \forall k \in \{1, \dots, K\},\forall l \in \{1,\ldots,|S_{\cal R}|\}
\end{align}
which is equivalent to Eq. \ref{Eq:RCPSP_constraint3} when the operand net has been constructed via Algs. \ref{Alg:AoN-to-ACT} and \ref{Alg:ACT-to-OperandNet}.  Note that the ${M}_{l}^+U_{l}^+[k]$ terms do not affect on the nonnegativity constraint because they only have a positive effect on the state transition function.   
\item For Eq. \ref{Eq:RCPSP_constraint4}, it is guaranteed by the state transition function in Eq. \ref{Eq:OperandNet-STF1_HFNMCF_RCPSP}, the initial condition in Eq. \ref{Eq:InitCond_HFNMCF_RCPSP}, and the non-negativity lower bound in Eq. \ref{Eq:Capacity_HFNMCF_RCPSP}.  Recognizing that in the non-renewable operands case, ${M}^+_{{\cal lj}l}=0 \forall j \in \{1, \dots, |{\cal E}_l|\},\forall l \in \{1,\ldots,|S_{\cal R}|\}$, Eq. \ref{Eq:RCPSP_constraint4} collapses to the non-negativity constraint on $Q_{{\cal l}l}[k] \forall l \in \{1,\ldots,|S_{\cal R}|\}$.
\item Taking account of all of the constraints in Eqs. \ref{Eq:OperandNet-STF1_HFNMCF_RCPSP} - \ref{Eq:Capacity_HFNMCF_RCPSP} reveals that there are no additional (unused) constraints that could potentially yield an optimal solution on $U^{-*}_{l}[k]$ that is more constrained than the optimal solution ${\cal x}_{{\cal i}k}^*$ provided by Eqs. \ref{Eq:RCPSP_Objective}--\ref{Eq:RCPSP_constraint4}.  
\item Because ${\cal x} \subset X$, Eqs. \ref{Eq:ObjFunc_HFNMCF_RCPSP}--\ref{Eq:Capacity_HFNMCF_RCPSP} are a formal generalization of Eqs. \ref{Eq:RCPSP_Objective}--\ref{Eq:RCPSP_constraint4}.  
\end{itemize}
\end{proof}
\noindent Consequently, Thm. \ref{Thm:HFNMCF-RCPSP} proves that the RCPSP specialization of the HFNMCF problem can be used for at least all the problem instances of these two variants of the RCPSP.  

Thm. \ref{Thm:HFNMCF-RCPSP} proves that the RCPSP specialization of the HFNMCF problem is more general than the RCPSP because ${\cal x} \subset X$.  While it may perhaps appear computationally undesirable to have more (integer) decision variables $X$ (in the RCPSP specialization of the HFNMCF problem) than the decision variable ${\cal x}$ (in the RCPSP), in reality, there is little to no practical impact on the computational performance of one optimization program over the other.  The results presented in Sec. \ref{Sec:Solution_HFNMCF_RCPSP} showed that $ U^{-}_{xl}[k] \;\equiv\; {\cal x}_{{\cal i}k} \forall x,k$, $U_l^{+}[k]$ is stated implicitly in Eq. \ref{Eq:RCPSP_constraint2}, $Q^{*}_{{\cal E}l}[k]$ is stated implicitly in the quantity $\sum_{\kappa=k-d_i}^{\kappa=k} {\cal x}_{{\cal i}\kappa} \forall k \in \{1, \dots, K\}$, and $Q_{S_l}[k]$ is stated implicitly in the slack variables associated with Eq. \ref{Eq:RCPSP_constraint3}.  Consequently, an efficient solution algorithm for the RCPSP specialization of the HFNMCF would collapse these explicit decision variables into implicit statements as part of a pre-solve routine, incurring negligible computational expense.  

Meanwhile, the inclusion of these explicit decision variables means that one of the primary advantages of the RCPSP specialization of the HFNMCF problem is that it has an \emph{explicit} and \emph{complete} description of the project state, including: 
\begin{enumerate}
\item when project activities are started in $U_{l}^{-}$, 
\item when project tasks are ongoing $Q^{*}_{{\cal E}l}$, 
\item when project tasks are completed $U_{l}^{+}$, and 
\item when the preconditions for project activities have not been fully met $Q_{S_l}$ -- leading to wait times between activities.   
\end{enumerate}
This complete description of the project state provides enhanced opportunities for project monitoring and control.   For example, it may be necessary to place exogenous ``blackout periods" on $Q^{*}_{{\cal E}l}$ where a given task may not occur.  Similarly, it may be necessary to place exogenous deadlines on the completion times of certain project activities in $U_{l}^{+}$ (rather than the project as a whole).  Moreover, it may be necessary to place ``use it lose it" expiration constraints on specific operands through $Q_{S_l}$.  Such real-world extensions of the RCPSP specialization of the HFNMCF problem are more easily introduced with an explicit description of the project state than when it remains implicit.  Finally, the explicit description of the project state lends itself to integration with well-established project monitoring and control methods, such as earned value analysis (EVA) \cite{hazir:2015:00, das:2025:00} and the earned schedule \cite{vandevoorde:2006:00}.  For example, the schedule variance (SV) compares the (actually) earned value (EV or budgeted cost of work performed) to the planned value (PV or the budgeted cost of work scheduled).
\begin{align}
SV &= EV - PV \\ 
   &= \sum_{k=1}^{k=K} V_{U^+}^TU_{l}^{+\ddagger}[k] - V_{U^+}^TU_{l}^{+*}[k]
\end{align}
where it is assumed that value only accrues when a project activity has been \emph{completed}, $V_{U^+}$ is the value vector associated with each project activity, the superscript $()^\ddagger$ denotes the project activity as it is \emph{actually} executed, and the superscript $()^*$ denotes the project as it is \emph{optimally} planned.  Similarly, the scheduled performance index is $ SPI = EV/PV$.  In all, the explicit description of the project state provides a straightforward basis for more advanced methods of project monitoring and control.  

Another advantage of the RCPSP specialization of the HFNMCF problem is that it can describe both the renewable operand and non-renewable operand cases with the \emph{same} optimization program formulation (albeit with different parameter values).  Fig. \ref{fig:OperandNet} shows that the case of renewable operands has return arcs back to an operand net place, while the case of non-renewable operands does not.  The structure of the RCPSP specialization for the HFNMCF problem remains unchanged.  Instead, the topology of the operand net, as exhibited by its positive and negative incidence matrices $M_l^+, M_l^-$ as constant parameter matrices, changes.  This means the same solution algorithm may be used for both variants of the problem.  Perhaps more importantly, the RCPSP specialization of the HFNMCF problem shows that the RCPSP taxonomy, where $\alpha_3=1$ (for renewable operands), $\alpha_3=T$ (for non-renewable operands), $\alpha_3=1T$ (for a mix of renewable and non-renewable operands), is either inadequate or unnecessary.  One may consider the $\alpha_3$ taxonomy inadequate because one may conceive many more operand net topologies that govern the constrained evolution of operands than what can be described by the $\alpha_3 \in \{o, 1, T, 1T, v, \chi, \sigma\}$ parameter\cite{demeulemeester:2002:00, Brucker:2007:00}.  In another respect, the $\alpha_3$ taxonomy may be considered unnecessary because the operand net and its state transition function provide a higher-level systems thinking abstraction that describes a family of potential operand behaviors that is more descriptive than what can be achieved by a taxonomy of parameter values (e.g., $\alpha_3$).  Indeed, a recent review of the RCPSP (and its many variants) \cite{artigues:2025:00} recognizes that there is a ``scheduling zoo" brought on by the current taxonomy, and there may exist a need to overcome it.  The reconciliation presented in Sec. \ref{Sec:Reconciliation} shows that a relatively small number of definitions govern the  RCPSP specialization of the HFNMCF problem and that many branches of the RCPSP taxonomy can be accommodated flexibly within the topology of the operand net.  

Perhaps of even greater interest is the potential application of the (general) HFNMCF problem (in Eq. \ref{Eq:HFNMCF_ObjFunc}-\ref{Eq:HFNMCF_DeviceModels2}) to the scheduling of mega-projects.  In a recent review, Denicol et al. emphasize that mega-projects stand up dedicated enterprises for carrying out projects of ever-greater scale\cite{denicol:2020:00, Hosseini:2025:ISC-J55}.  In this respect, the RCPSP taxonomy is noticeably silent.  The imposition of no structural resources ($\alpha_1=\emptyset$) throughout the entirety of the RCPSP taxonomy renders it unable to characterize the nature of these enterprises and the complexities they bring to the planning and execution of mega-projects.  The broader machine scheduling taxonomy\cite{demeulemeester:2002:00,Brucker:2007:00}, which allows for structural resources, is therefore more relevant to the scheduling of mega-projects.   However, the proliferation of a parameter-based machine scheduling taxonomy may prove to be unwieldy, just as the parameter-based RCPSP taxonomy was found to be inadequate or unnecessary in the previous paragraph.  Instead, mega-project enterprises are likely to be characterized by a wide variety of formal and functional structures and aggregations that require a plurality of higher-level systems thinking abstractions.  In this respect, and as highlighted in Sec. \ref{Sec:MBSEIntro}, MBSE offers a promising approach to model, analyze, and operationalize enterprise architectures for mega-project management\cite{zhang2025mbse,sitton:2018:00, atencio:2022:00, Hosseini:2025:ISC-J55}.  From there, the SysML block definition diagrams (BDDs) and activity diagrams (ACTs) can be translated with HFGT into the HFNMCF problem.  

The HFNCMCF problem in Eqs. \ref{Eq:HFNMCF_ObjFunc}-\ref{Eq:HFNMCF_DeviceModels2} has several distinguishing features that make it particularly promising for the scheduling of mega-projects.  First, it distinguishes between the state evolution of the operand (i.e., project) (in Eq. \ref{Eq:HFNMCF_OperandNet-STF1},\ref{Eq:HFNMCF_OperandNet-STF2}) from the state evolution of the engineering system (i.e., enterprise) (in Eq. \ref{Eq:HFNMCF_ESN-STF1}, \ref{Eq:HFNMCF_ESN-STF2}) tasked with evolving that project forward.  It recognizes that the states of the operand $(Q_{Sl_i}[k],Q_{{\cal E}l_i}[k],U_{l_i}^+[k],U_{l_i}^-[k])$ are distinct from the states of the engineering system $(Q_{B}[k],Q_{\cal E}[k],U^+[k],U^-[k])$ although they are synchronized in a many-to-many relationship (in Eqs. \ref{Eq:HFNMCF_SyncPlus},\ref{Eq:HFNMCF_SyncMinus}).  One may envision a mega-project enterprise that possesses numerous capabilities (as in Defn. \ref{Defn:Capability}) that appear as multiple engineering system net transitions that are all synchronized with a specific operand net transition.  This plurality of capabilities can occur because the enterprise has invested in multiple resources (as in Defn. \ref{Defn:SystemResource}) that can carry out the same process (e.g., two workers with the same skill set).  Alternatively, the plurality of capabilities can occur because the enterprise has two different processes (or methods) for completing the same project task (e.g., paper-based processing vs. software-based processing).  Additionally, one may envision a mega-project enterprise that has multiple layers of centralized, distributed, and decentralized decision-making capabilities (in the engineering system net) that are required as preconditions for the (value-adding) capabilities that actually evolve a project forward.  Indeed, in many mega-projects, it is these decision-making capabilities (e.g., licensing and permitting) that most affect the project's ultimate delivery time.  Naturally, each of these capabilities is likely to have durations for its execution.  Consequently, the HFNMCF problem recognizes that the duration of an operand net transition (in Eq. \ref{Eq:HFNMCF_OperandNetDurationConstraint}) (e.g., the expiration of perishable goods) is distinct from the duration of an engineering system net transition (in Eq. \ref{Eq:HFNMCF_DurationConstraint}) (e.g., the delivery time of perishable goods).  Additionally, it recognizes that engineering systems and their operands, by default, are open systems with potentially many exogenous constraints (in Eqs. \ref{Eq:HFNMCF_ESN-Exogenous},\ref{Eq:HFNMCF_OperandNet-Exogenous}).  In all, the HFNCMCF problem presents ample opportunities to explore the scheduling of mega-projects, taking into account the complexities of the mega-project's enterprise architecture.

\section{Conclusion}\label{Sec:Conclusion}
This paper reconciles the resource-constrained project scheduling problem (RCPSP) with model-based systems engineering (MBSE) and hetero-functional graph theory (HFGT) by (i) constructing a concrete translation pipeline from an Activity-on-Node (AoN) network to a SysML activity diagram (ACT) and then to an operand net (Algs.~\ref{Alg:AoN-to-ACT},\ref{Alg:ACT-to-OperandNet}), and (ii) specializing the hetero-functional network minimum-cost flow (HFNMCF) problem to the RCPSP context, as the systematic means of HFGT for quantitative analysis. These steps establish a common, state-based language that makes the relationships among RCPSP, MBSE, and HFGT precise. The theoretical result proves that the RCPSP specialization of the HFNMCF problem is a \emph{formal generalization} of the RCPSP. The proof aligns indices and decision variables, shows objective equivalence, and demonstrates that the operand-net state transition, duration, initial/final, and capacity conditions \emph{necessitate} the RCPSP precedence and resource relations. Since ${\cal x}\subset X$, every RCPSP instance is representable and solvable as an HFNMCF specialization, while the specialization admits real-world extensions not expressible in the classic RCPSP without altering the model structure.

The other key methodological contribution is the \emph{explicit and complete} project-state description delivered by the operand-net variables: (i) operand net input firing vector, (ii) operand net transition marking vector, (iii) operand net output firing vector, and (iv) operand net place marking vector. The numerical results verify that optimum operand net input firing vector is equivalent to RCPSP optimum decision variables (activity start times) and that operand net transition marking vector($Q^{*}_{{\cal E}l}$) and operand net place marking vector ($Q_{S_l}$) recover, respectively, the task-on-going profile and inter-activity slack that RCPSP only encodes \emph{implicitly}. This explicit state directly supports richer monitoring and control and integrates naturally with earned value and earned schedule measures defined on operand net output firing vector ($U_{l}^{+}$).

From the computational efficiency perspective, while the specialization introduces more explicit integer variables than the classic RCPSP, efficient pre-solve can collapse operand net output firing vectors ($U_{l}^{+}$), operand net transition marking vectors ($Q^{*}_{{\cal E}l}$), and operand net place marking vectors ($Q_{S_l}$) to implicit expressions when desired, preserving competitiveness with standard RCPSP solvers while retaining the ability to ``turn on'' explicit state for monitoring and control stage. From the structural perspective, the same optimization program captures both renewable and non-renewable operands; the distinction is carried by the operand-net topology and operand net incidence matrices. This unification renders portions of the traditional $\alpha_3$ taxonomy unnecessary for modeling operand behavior. Rather than proliferating parameter classes in a ``scheduling zoo'', the operand-net abstraction compacts many variants into topology and state-transition semantics, enabling one solver and one model interface to cover a much broader family of constraints through data edits rather than reformulation.

Beyond single projects, the general HFNMCF program is naturally suited to mega-projects, while the classic RCPSP is silent due to the imposition of no structural resources. HFNMCF program separates (and synchronizes) the state evolution of the \emph{operand} (the project) from that of the \emph{engineering system} (the dedicated enterprise executing it). This separation captures several capabilities, alternative methods, layered decision processes in the mega-project enterprise, and exogenous constraints typical of enterprise architectures. The MBSE specifies these structures at the design level; HFGT compiles them into a quantitative framework, creating a coherent digital thread from architecture to executable schedule and decisions as previously claimed in \cite{Hosseini:2025:ISC-J55}. In summary, by proving equivalence and demonstrating generalization, this work reframes the RCPSP as a special case of a broader HFNMCF model that preserves RCPSP's strengths while extending it for enabling richer monitoring, unifying operand behaviors, and opening a principled path from enterprise architecture to executable, detail-oriented decisions and plans for delivery of complex megaprojects.

\printnomenclature

\appendix 
For the case of non-renewable resources, the RCPSP specialization of HFNMCF problem found in Eq. \ref{Eq:ObjFunc_HFNMCF_RCPSP}--\ref{Eq:Capacity_HFNMCF_RCPSP} is solved to optimality.  The optimal values of the column vectors $U_l^{-*}[k]$ and $U_l^{+*}[k]$ are horizontally concatenated and presented below.  
\begin{footnotesize}
\begin{align}\label{Eq:OptimalU2}
U_{l}^{-*} =
\left[
\begin{array}{*{17}{r|}r}
1 &   &   &   &   &   &    &    &    &    &    &    &    &    &    &    &   &  \\\hline
1 &   &   &   &   &   &    &    &    &    &    &    &    &    &    &    &   &  \\\hline
  &   & 1 &   &   &   &    &    &    &    &    &    &    &    &    &    &   &  \\\hline
1 &   &   &   &   &   &    &    &    &    &    &    &    &    &    &    &   &  \\\hline
  &   &   &   &   & 1 &    &    &    &    &    &    &    &    &    &    &   &  \\\hline
  &   &   &   & 1 &   &    &    &    &    &    &    &    &    &    &    &   &  \\\hline
  &   &   &   &   &   &    & 1  &    &    &    &    &    &    &    &    &   &  \\\hline
  &   &   &   &   &   &    &    &    &    &    & 1  &    &    &    &    &   &  \\\hline
  &   &   &   &   &   &    &    &    &    &    &    &    & 1  &    &    &   &  \\
\end{array}
\right],
U_{l}^{+*} =
\left[
\begin{array}{*{17}{r|}r}
  &   & 1 &   &   &   &    &    &    &    &    &    &    &    &    &    &   &  \\\hline
  &   &   &   &   &   &    & 1  &    &    &    &    &    &    &    &    &   &  \\\hline
  &   &   &   &   & 1 &    &    &    &    &    &    &    &    &    &    &   &  \\\hline
  &   &   &   & 1 &   &    &    &    &    &    &    &    &    &    &    &   &  \\\hline
  &   &   &   &   &   &    &    &    &    &    &    &    &  1 &    &    &   &  \\\hline
  &   &   &   &   &   &    &    &    &    &  1 &    &    &    &    &    &   &  \\\hline
  &   &   &   &   &   &    &    &    &    &    &  1 &    &    &    &    &   &  \\\hline
  &   &   &   &   &   &    &    &    &    &    &    &    &  1 &    &    &   &  \\\hline
  &   &   &   &   &   &    &    &    &    &    &    &    &  1 &    &    &   &  \\
\end{array}
\right]
\end{align}
\end{footnotesize}
\noindent Additionally, the optimal value $Q^{*}_{{\cal E}l[k]}$ can be calculated after the optimization via Eq. \ref{Eq:HFNMCF_OperandNet-STF2}.  
\begin{footnotesize}
\begin{align}
Q^{*}_{{\cal E}l}[k] =
\left[
\begin{array}{*{18}{r|}r}
1 & 1 &   &   &   &   &    &    &    &    &    &    &    &    &    &    &   &  &   \\\hline
1 & 1 & 1 & 1 & 1 & 1 & 1  &    &    &    &    &    &    &    &    &    &   &  &   \\\hline
  &   & 1 & 1 & 1 &   &    &    &    &    &    &    &    &    &    &    &   &  &   \\\hline
1 & 1 & 1 & 1 &   &   &    &    &    &    &    &    &    &    &    &    &   &  &   \\\hline
  &   &   &   &   & 1 & 1  & 1  & 1  & 1  & 1  & 1  & 1  &    &    &    &   &  &   \\\hline
  &   &   &   & 1 & 1 & 1  & 1  & 1  & 1  &    &    &    &    &    &    &   &  &   \\\hline
  &   &   &   &   &   &    & 1  & 1  & 1  & 1  &    &    &    &    &    &   &  &   \\\hline
  &   &   &   &   &   &    &    &    &    &    & 1  & 1  &    &    &    &   &  &   \\\hline
  &   &   &   &   &   &    &    &    &    &    &    &    &    &    &    &   &  &   \\
\end{array}
\right]
\end{align}
\end{footnotesize}
\noindent Finally, the optimal values of $Q^{*}_{Sl}[k]$ are horizontally concatenated and presented below.  
\begin{footnotesize}
\begin{align}\label{Eq:OptimalQ2}
Q_{S_l} =
\left[
\begin{array}{*{18}{r|}r}
3 &   &   &   &   &   &   &   &   &   &   &   &   &   &   &   &   &   &  \\\hline
25& 16& 16& 12& 12& 10& 7 & 7 & 4 & 4 & 4 & 4 &   &   &   &   &   &   &   \\\hline
  &   &   &   &   &   &   &   &   &   &   &   &   &   &   &   &   &   &   \\\hline
  &   &   &   &   & 1 & 1 &   &   &   &   &   &   &   &   &   &   &   &   \\\hline
  &   &   &   &   &   &   &   &   &   &   &   &   &   &   &   &   &   &   \\\hline
  &   &   &   &   &   &   &   &   &   &   &   &   &   &   &   &   &   &   \\\hline
  &   &   &   &   &   &   &   &   &   &   &   &   &   &   &   &   &   &   \\\hline
  &   &   &   &   &   &   &   &   &   &   &   &   &   &   &   &   &   &   \\\hline
  &   &   &   &   &   &   &   &   &   &   &   &   &   &   &   &   &   &   \\\hline
  &   &   &   &   &   &   &   &   &   &   &   &   &   &   &   &   &   &   
\end{array}
\right]
\end{align}
\end{footnotesize}

\bibliographystyle{IEEEtran}
\bibliography{LIINESLibrary,LIINESPublications,RCPSP_library}

\end{document}